\documentclass[journal]{IEEEtran}
\usepackage{amsmath,amsthm,graphicx,amssymb,cite,enumerate,empheq, mathtools,multirow,float,caption,subcaption, tabularx, booktabs, threeparttable, color,xcolor, hyperref}
\mathtoolsset{showonlyrefs} %
\theoremstyle{plain}

\usepackage{algorithm}%
\usepackage{algpseudocode}

\newtheorem{proposition}{Proposition}[section]

\newtheorem{cor}{Corollary}
\theoremstyle{remark}
\newtheorem{rem}{Remark}
\newtheorem{fact}{Fact}
\floatstyle{plaintop}
\restylefloat{table}
\newcommand{\ip}[2]{\left\langle #1,\,#2\right\rangle}
\newcommand{\myOp}[1]{\mathbf{#1}}
\newcommand{\mySetOp}[1]{#1}
\newcommand{\reply}[1]{#1}
\newcommand{\replySec}[1]{#1}
\newcommand{\replyThird}[1]{#1}
\algnewcommand\algorithmicinput{{\bfseries Input}}%
\algnewcommand\algorithmicoutput{{\bfseries Output}}%
\algnewcommand\AlgInput{\item[\algorithmicinput]}%
\algnewcommand\AlgOutput{\item[\algorithmicoutput]}%
\algrenewcommand\Return{\State\textbf{return} }%

\ifCLASSINFOpdf
\else
\fi
\begin{document}
\bstctlcite{IEEEexample:BSTcontrol}
\title{Convergent Primal-Dual Plug-and-Play Image Restoration: A General Algorithm and Applications}

\author{Yodai~Suzuki,~\IEEEmembership{Student Member,~IEEE,}
        Ryosuke~Isono,~\IEEEmembership{Graduate Student Member,~IEEE,}
        and Shunsuke~Ono,~\IEEEmembership{Senior Member,~IEEE}% <-this % stops a space
%\thanks{Manuscript received XXX, XXX; revised XXX XXX, XXX.}
\thanks{Y. Suzuki is with the Department of Computer Science, Institute of Science Tokyo, Yokohama, 226-8503, Japan (e-mail: suzuki.y.bd8c@m.isct.ac.jp).}% <-this % stops a space
\thanks{R. Isono is with the Department of Computer Science, Institute of Science Tokyo, Yokohama, 226-8503, Japan (e-mail: isono.r.1f44@m.isct.ac.jp).}% <-this % stops a space
\thanks{S. Ono is with the Department of Computer Science, Institute of Science Tokyo, Yokohama, 226-8503, Japan (e-mail: ono@comp.isct.ac.jp).}% <-this % stops a space
%This work was supported in part by JST FOREST under Grant JPMJFR232M and JST AdCORP under Grant JPMJKB2307, and in part by JSPS KAKENHI under Grant 22H03610, 22H00512, 23H01415, 23K17461, 24K03119, 24K22291, 25H01296, and 25K03136.

\thanks{This work was supported in part by JST FOREST under Grant JPMJFR232M, JST AdCORP under Grant JPMJKB2307, JST CREST under Grant JPMJCR25Q5, JST ACT-X under Grant JPMJAX24C1, JST BOOST, Japan Grant Number JPMJBS2430, and in part by JSPS KAKENHI under Grant 22H03610, 22H00512, 23H01415, 23K17461, 24K03119, 24K22291, 25H01296, and 25K03136.
%by JST PRESTO under Grant JPMJPR21C4, JST AdCORP under Grant JPMJKB2307, JST ACT-X under Grant JPMJAX24C1, JST BOOST, Japan Grant Number JPMJBS2430, and in part by JSPS KAKENHI under Grant 22H03610, 22H00512, 23H01415, 23K17461, 24K03119, and 24K22291.
}
\thanks{The code is available at \url{https://github.com/MDI-TokyoTech/Convergent_Primal-Dual_Plug-and-Play_Image_Restoration_A_General_Algorithm_and_Applications}.}
%\thanks{Y. Suzuki is with the Department of Computer Science, Tokyo Institute of Technology, Yokohama, 226-8503, Japan (e-mail: suzuki.y.du@m.titech.ac.jp).}
%\thanks{R. Isono is with the Department of Computer Science, Tokyo Institute of Technology, Yokohama, 226-8503, Japan (e-mail: isono.r.ac@m.titech.ac.jp).}
%\thanks{S. Ono is with the Department of Computer Science, Tokyo Institute of Technology, Yokohama, 226-8503, Japan (e-mail: ono@c.titech.ac.jp).}
% <-this % stops a space
}

\markboth{IEEE TRANSACTIONS ON COMPUTATIONAL IMAGING}%
{Shell \MakeLowercase{\textit{et al.}}: Bare Demo of IEEEtran.cls for IEEE Journals}



\maketitle

\begin{abstract}
We propose a general deep plug-and-play (PnP) algorithm with a theoretical convergence guarantee. PnP strategies have demonstrated outstanding performance in various image restoration tasks by exploiting the powerful priors underlying Gaussian denoisers. However, existing PnP methods often lack theoretical convergence guarantees under realistic assumptions due to their ad-hoc nature, resulting in inconsistent behavior. Moreover, even when convergence guarantees are provided, they are typically designed for specific settings or require a considerable computational cost in handling non-quadratic data-fidelity terms and additional constraints, which are key components in many image restoration scenarios. To tackle these challenges, we integrate the PnP paradigm with primal-dual splitting (PDS), an efficient proximal splitting methodology for solving a wide range of convex optimization problems, and develop a general convergent PnP framework. Specifically, we establish theoretical conditions for the convergence of the proposed PnP algorithm under a reasonable assumption. Furthermore, we show that the problem solved by the proposed PnP algorithm is not a standard convex optimization problem but a more general monotone inclusion problem, where we provide a mathematical representation of the solution set. Our approach efficiently handles a broad class of image restoration problems with guaranteed theoretical convergence. Numerical experiments on specific image restoration tasks validate the practicality and effectiveness of our theoretical results.
\end{abstract}

\begin{IEEEkeywords}
Image restoration, plug-and-play (PnP) algorithms, primal-dual splitting (PDS), convergence guarantee.
\end{IEEEkeywords}

% \fi
\IEEEpeerreviewmaketitle

\section{Introduction}

% image restoration
\subsection{Image Restoration by Proximal Splitting Algorithms}
\IEEEPARstart{I}{mage} restoration is a longstanding and essential problem with diverse applications, ranging from remote sensing, geoscience, and astronomy to biomedical imaging. In general, it is reduced to the inverse problem of estimating an original image from an observed image, which is degraded by some linear observation process and noise contamination (e.g., additive white Gaussian noise and Poisson noise). 

Since this inverse problem is often ill-posed or ill-conditioned, a standard approach formulates the restoration task as a convex optimization problem to characterize a desirable solution. In this scheme, we typically minimize the sum of a data-fidelity term with respect to a given observed image and a regularization term reflecting prior knowledge on natural images. Typically, the data-fidelity is chosen as the negative log-likelihood of the statistical distribution assumed for the noise contamination. Conversely, the choice of the regularization term is highly diverse. One of the most popular techniques is Total Variation (TV)~\cite{TV_original} and its extensions\cite{TV1,TVExtension,TV6,TV5}, which promote the spatial piecewise smoothness of images. In any case, these data-fidelity and regularization terms are possibly nonsmooth, which makes it difficult to solve the convex optimization problem analytically or using gradient-based methods.

This difficulty has been overcome by proximal splitting algorithms, which assume that the proximal operators associated with these terms can be computed efficiently~\cite{ProximalSplittingAlgorithm1,ProximalSplittingAlgorithm2,ProximalSplittingAlgorithm3}. Notable examples include Forward-Backward Splitting (FBS)~\cite{FBS}, Alternating Direction Method of Multipliers (ADMM)~\cite{ADMM}, and Primal-Dual Splitting (PDS)~\cite{PDS1, PDS4, PDS2, PDS3}. Among these algorithms, PDS offers significant flexibility and efficiency, as it can handle nonsmooth convex terms involving linear operators without requiring matrix inversions or performing inner iterations, which is generally impossible in the case of FBS or ADMM. These features offer several advantages for image restoration applications~\cite{PDSApplication1, PDSApplication2} (refer to Section~\ref{ApplicationToImageRestoration} for more detailed explanation with specific formulations).

\subsection{Plug-and-Play (PnP) Algorithms}
\begin{table*}[t]
    \begin{center}
        \caption{Pros and Cons of Existing PnP Methods and the Proposed Method.}
        \label{PnPmethodsComparison}
            \scalebox{0.73}{
                \begin{tabular}{llcccccc} 
                    \toprule
                   Algorithms & Authors &  Data-fidelity &  
                   \begin{tabular}{c} Additional\\ constraints\end{tabular}  & 
                   \begin{tabular}{c} Inversion\\ free\end{tabular}
                   & \begin{tabular}{c} Convergence\\guarantees \end{tabular} &
                    \begin{tabular}{c} Solution\\ characterization\end{tabular}& 
                    \begin{tabular}{c} Assumption\\ on denoiser $J$ \end{tabular}
                    \vspace{0mm} \\ \midrule 
                   \multirow{3}{*}{PnP-FBS} & \replySec{Sun et.al.~\cite{Sun2019TCI}} & \replySec{Quadratic} & \replySec{$-$} & \replySec{\checkmark} & \replySec{\checkmark} & \replySec{$-$} & \replySec{$J$ is averaged.} \\  
                    & Pesquet et.al.~\cite{Pesquet2021} & Quadratic & $-$ & \checkmark & \checkmark & \checkmark &  $J$ is firmly nonexpansive. \\
                   & Ebner et al.~\cite{PnPisConvergent} & Quadratic &  $-$ & \checkmark &  \checkmark & \checkmark & $J$ is contractive.
                    \\  \midrule 
                   \multirow{6}{*}{PnP-ADMM} 
                    & Venkatakrishnan et al.~\cite{PnP0} & Any$^\text{\reply{*}}$ & $-$ & $-$ & $-$ & $-$ & $-$\\
                    & Rond et al.~\cite{PoissonPnP} & GKL div. & $-$ & $-$ & $-$ & $-$ & $-$\\
                    & Sreehari et al.~\cite{ADMMPnPConvergence1} & Any$^\text{\reply{*}}$ & $-$ & $-$  & \checkmark & $-$ & $\nabla J$ is a symmetric matrix, etc. \\
                    & S. H. Chan et al.~\cite{ADMMPnPConvergence2} & Any$^\text{\reply{*}}$ & $-$ & $-$  & \checkmark & $-$ & $J$ is bounded, etc. \\
                    & P. Nair et al.~\cite{IstaAndADMMPnPConvergence} & Any$^\text{\reply{*}}$ & $-$ & $-$  & \checkmark & $-$ & $J$ is firmly nonexpansive. \\ 
                    & Sun et al.~\cite{ConvergentPnPADMM} & Any$^\text{\reply{*}}$ & \checkmark & $-$ & \checkmark & \checkmark & $J$ is firmly nonexpansive. \\ \midrule 
                   \multirow{3}{*}{PnP-PDS} & Ono\cite{PnPPDS}& Any$^\text{\reply{*}}$ & \checkmark & \checkmark & $-$ & $-$ & $-$ \\
                   & Garcia et al.~\cite{PnPPDSByPesquet} & Constraint & \checkmark & \checkmark & \reply{\checkmark} & \reply{\checkmark} & $J$ is firmly nonexpansive.\\
                   & Ours & Any$^\text{\reply{*}}$ & \checkmark & \checkmark &\checkmark & \checkmark & $J$ is firmly nonexpansive.\vspace{-1mm}\\
                   \bottomrule 
                   \multicolumn{8}{l}{\small{$^{*}$ \reply{Data-fidelity functions should be proper lower-semicontinuous convex, and their proximal operators should have closed-form expressions or be efficiently computable.}}}
               \end{tabular}
            }
    \end{center}
\end{table*}

Recently, regularization using Gaussian denoisers has become an active area of research. This line of work was pioneered by a PnP paradigm that was proposed in~\cite{PnP0} and later extended, e.g., in~\cite{PnPLikeWithBM3D, PnPLikeWithBM3DExtended, PoissonPnP}. The basic idea of PnP approaches is to replace the proximity operator of the regularization term used in proximal splitting algorithms with an off-the-shelf Gaussian denoiser, such as Block-Matching and 3D Filtering (BM3D)~\cite{BM3D}, Non-Local Means (NLM)~\cite{NLM}, Trainable Nonlinear Reaction Diffusion (TNRD)~\cite{TNRD}, and Denoising Convolutional Neural Networks (DnCNN)\cite{DnCNN1}. Numerous studies have shown that the PnP strategy achieves superior performance across a variety of applications, including tomographic imaging~\cite{PnP0, ADMMPnPConvergence1}, biomedical imaging~\cite{PnPBio}, magnetic resonance imaging~\cite{PnP1}, and remote sensing~\cite{PnPRemoteSensing} (see~\cite{PnP3} for a recent review). 

Following these impressive results, several works explore attempts to plug in more sophisticated denoisers. Zhang et al. proposed a neural network specifically designed for use in PnP methods~\cite{DRUNet}, known as DRUNet. Duff et al. proposed \textit{generative regularizers}~\cite{GenerativePnP}, realized by Autoencoders (AE), Variational Autoencoders (VAE)~\cite{VAE}, and Generative Adversarial Networks (GAN)~\cite{GAN}. Zhu et al. showed that diffusion models~\cite{DiffusionModel} also have promising performance as prior knowledge used in PnP methods~\cite{DiffusionPnP}. We should mention that despite such extensive research, the theoretical convergence guarantee for PnP methods remains a challenging open question~\cite{RealSN, ADMMPnPConvergence1, ADMMPnPConvergence2, Sun2019TCI, PnPisConvergent}. This is mainly because replacing the proximity operator of the regularization term with a denoiser does not necessarily guarantee the existence of the underlying convex optimization problem.

To explore explicit data-driven regularization, alternative frameworks including regularization by denoising (RED)~\cite{RED,REDPRO}, gradient step denoisers (GSD)~\cite{GSD_Cohen, GSD_Hurault,GSD_Hurault2,GSD_Quasi}\reply{, and its generalization in Bregman geometry~\cite{BregmanProximalGradient}} have been proposed. Unlike PnP strategies, they explicitly define a regularization term derived from the output of a denoiser, incorporating it directly into the formulation of the optimization problem. Subsequently, they construct algorithms to solve this problem as a special case of traditional proximal splitting algorithms. Nonetheless, it should be noted that they still rely on conditions that may not be \reply{desirable} in real-world settings. For example, RED imposes strict assumptions on the denoiser to guarantee its convergence, which are often impractical for deep denoisers~\cite{REDAnalysis} (see Section~\ref{ComparisonWithRED} for details). The GSD approach introduces a non-convex regularization term, \replySec{which makes it difficult to guarantee convergence to globally optimal solutions and may cause the final solution to depend on the initialization}.

% Convergence of PnP 
In contrast, within the context of PnP methods, some studies have been conducted to establish their convergence under realistic assumptions, focusing on the firm nonexpansiveness of denoisers. Pesquet et al. showed theoretically that a subclass of maximally monotone operators (MMO) can be fully represented by the firmly nonexpansive DnCNN with a certain network structure. They also constructed a PnP method based on forward-backward splitting (PnP-FBS) with convergence guarantees for image restoration~\cite{Pesquet2021}. \reply{As highlighted by the authors, their technique can be extended to any algorithm, as long as its analysis relies on the theory of MMOs, including PDS.} Sun et al. proposed a PnP based on ADMM (PnP-ADMM) with the convergence guarantee by plugging in a firmly nonexpansive denoiser~\cite{ConvergentPnPADMM}. \reply{Hertrich et al.~\cite{HERTRICH2021203} developed a method to train averaged CNNs by stochastic gradient descent on the Stiefel manifold.} %It is worth mentioning that both of these works also provide a mathematical characterization of limit points of sequences generated by PnP algorithms.

% PnP-PDS
Regarding the flexibility of PnP algorithms, PnP based on PDS (PnP-PDS) offers the advantage of efficiently handling a wider range of image restoration problems than PnP-FBS and PnP-ADMM, as mentioned above. Ono proposed the first PnP-PDS and demonstrated its practicality, while no theoretical convergence guarantees were provided~\cite{PnPPDS}. On the other hand, Garcia et al. proposed a convergent PnP-PDS with a firmly nonexpansive denoiser for a specific case and achieved successful results in the restoration of images obtained by a photonic lantern~\cite{PnPPDSByPesquet}. We remark that the denoiser assumption of firmly nonexpansiveness is reasonable in actual applications, as shown in~\cite{Pesquet2021}. Nevertheless, the convergence of PnP-PDS in its general form has not yet been fully analyzed. Once a general PnP-PDS with convergence guarantees is established, various image restoration problems can be approached with rich priors of Gaussian denoisers in an efficient and flexible fashion.

% Research question and contribution
\subsection{Contributions and Paper Organization}
Now, the following natural question arises: \textit{Can we construct a general PnP-PDS with theoretical convergence guarantees under realistic assumptions on denoisers?} In this paper, we answer this question by proposing a general convergent PnP-PDS strategy with a theoretical guarantee that has a wide range of practical applications. The main contributions of our paper are outlined below:\begin{itemize}
\setlength{\leftskip}{-10pt}
    \item We study a general PnP-PDS algorithm that extends the framework proposed in~\cite{PnPPDS} and establish its theoretical convergence conditions, building upon the foundational results presented in~\cite{Pesquet2021}.
    \item We show that the problem solved by our PnP-PDS algorithm is not a standard convex optimization problem but a certain monotone inclusion problem, and provide a mathematical representation of the solution set.
    \item \reply{Based on the above foundations, we show that our PnP-PDS efficiently solves typical inverse problems formulated with non-quadratic data-fidelity, with convergence guarantees. In particular, it allows the use of an $\ell_2$-ball data-fidelity constraint for image restoration under Gaussian noise, offering more flexibility in parameter selection and improving the versatility of the denoiser. Moreover, it effectively handles the generalized Kullback–Leibler (GKL) divergence for image restoration under Poisson noise.}
%    \item Based on the above theoretical foundations, we present specific results on the convergence properties of PnP-PDS to address typical image restoration problems, including non-quadratic data-fidelity terms and additional constraints.
 \end{itemize} 
In Table~\ref{PnPmethodsComparison}, we briefly summarize how the proposed method differs from related works (see Section~\ref{ComparisonWithOtherMethods} for detailed and comprehensive discussion). The advantages of the proposed convergent PnP-PDS are illustrated by experimental results on a number of typical image restoration tasks.

The remainder of this paper is organized as follows. In Section~\ref{sec:preliminaries}, we set up proximal tools, the PDS algorithm, and the firmly nonexpansive denoiser to construct our PnP-PDS. Section~\ref{sec:proposedmethod} presents our main results. We first establish a convergence guarantee for the general PnP-PDS algorithm and characterize its solutions. Subsequently, we introduce several specific forms of the algorithm designed for image restoration tasks, along with their corresponding solution characterizations. In Section~\ref{sec:experiments}, we conduct image restoration experiments, focusing on deblurring and inpainting tasks under two noise conditions: Gaussian noise and Poisson noise, and demonstrate that our PnP-PDS achieves state-of-the-art performance and high stability.

The preliminary version of this paper, without the applications to deblurring and inpainting under Poisson noise, the introduction of a box constraint, and the complete theoretical proof, can be found in the conference proceedings~\cite{PnPPDSOurs}.

\section{Preliminaries}
\label{sec:preliminaries}
\subsection{Notations}
%Vectors and matrices are denoted by bold lower and upper case letters, respectively. 
Let $\|\cdot\|_1$ and $\|\cdot\|_2$ denote the $\ell_1$-norm and $\ell_2$-norm, respectively. The transpose of a vector is denoted by $(\cdot)^\top$. The notation $\left[\cdot\right]_m$ indicates the $m$-th element of a vector. The adjoint of a bounded linear operator is denoted by $(\cdot)^*$, and the operator norm is represented as $\|\cdot\|_{\mathrm{op}}$.

Let $\mathcal{H}$ be a real Hilbert space with its inner product $\langle \cdot, \cdot \rangle$ and norm $\|\cdot\|$. The notation $\Gamma_0(\mathcal{H})$ represents the set of proper lower-semicontinuous convex functions from $\mathcal{H}$ to $\mathbb{R}\cup\{\infty\}$. For $\mathcal{F}\in\Gamma_0(\mathcal{H})$, its effective domain is $\mathrm{dom}(\mathcal{F}):=\{\mathbf{x}\in\mathcal{H}:\mathcal{F}(x)<\infty\}$. 

For $D\subset\mathcal{H}$, we consider an operator $T:D\to\mathcal{H}$. If $\|T\mathbf{x}_1-T\mathbf{x}_2\|\leq\|\mathbf{x}_1-\mathbf{x}_2\|$ holds for every $(\mathbf{x}_1,\,\mathbf{x}_2)\in D^2$, $T$ is said to be nonexpansive. When $2T-\mathrm{Id}$ is nonexpansive, $T$ is called firmly nonexpansive~\cite[Proposition~4.4]{ConvexBook} ($\mathrm{Id}$ denotes an identity operator). If $\kappa T$ is firmly nonexpansive for some $\kappa\in(0,\,\infty)$, then $T$ is $\kappa$-cocoercive.

Let $\mySetOp{A}:\mathcal{H}\to2^{\mathcal{H}}$ be a set-valued operator. Its graph is $\mathrm{gra}\mySetOp{A}=\{(\mathbf{x},\,\mathbf{y})\in\mathcal{H}^2:\mathbf{y}\in \mySetOp{A}\mathbf{x}\}$. If $\ip{\mathbf{x}_1-\mathbf{x}_2}{\mathbf{y}_1-\mathbf{y}_2}\geq0$ holds for every $(\mathbf{x}_1,\,\mathbf{y}_1)\in\mathrm{gra}\mySetOp{A}$ and $(\mathbf{x}_2,\,\mathbf{y}_2)\in\mathrm{gra}\mySetOp{A}$, $\mySetOp{A}$ is said to be monotone. The operator $\mySetOp{A}$ is called maximally monotone if there exists no monotone operator $\mySetOp{B}:\mathcal{H}\to2^{\mathcal{H}}$ such that $\mathrm{gra}\mySetOp{B}$ properly contains $\mathrm{gra}\mySetOp{A}$.

We introduce several specific convex functions and sets that we use in this paper. Let $C$ be a nonempty closed convex set on $\mathcal{H}$. The indicator function of $C$, denoted by $\iota_C$, is defined as
\begin{align}
\iota_C(\mathbf{x}):=
\begin{cases}
0,&\mathrm{if}\, \mathbf{x}\in C,\\
\infty,&\mathrm{otherwise}.
\end{cases}
\end{align}
Let $B_{2,\,\varepsilon}^{\mathbf{v}}$ be the $\mathbf{v}$-centered $\ell_2$-norm ball with radius $\varepsilon>0$, which is given by
\begin{align}
    B_{2,\,\varepsilon}^{\mathbf{v}}:=\left\{\mathbf{x}\in\mathcal{H}\mid\|\mathbf{x}-\mathbf{v}\|_2\leq\varepsilon\right\}.
\end{align}
The Generalized Kullback-Leibler (GKL) divergence between $\mathbf{x} \in \mathbb{R}^N$ and $\mathbf{v} \in \mathbb{R}^N$ is defined as
\begin{align}
&\hspace{-4mm}\mathrm{GKL}_\mathbf{v}(\mathbf{x})\\
&\hspace{-4mm}:=\displaystyle\sum_{i=1}^{N}\begin{cases}
\eta [\mathbf{x}]_i-[\mathbf{v}]_i\ln{\eta [\mathbf{x}]_i},\,&\hspace{-2mm}\text{if $[\mathbf{v}]_i>0$ and $[\mathbf{x}]_i>0$},\\
\eta [\mathbf{x}]_i,\,&\hspace{-2mm}\text{if $[\mathbf{v}]_i=0$ and $[\mathbf{x}]_i\geq0$,}\\
\infty,\,&\hspace{-2mm}\text{otherwise,}
\end{cases}
\label{definitionOfGKL}
\end{align}
where $\eta>0$ is a certain scaling parameter, which will be detailed in Section \ref{PoissonApplication}.

\subsection{Proximal Tools}
\label{ssec:prox}
For $\mathcal{F} \in\Gamma_0(\mathcal{H})$, the proximity operator of an index $\gamma>0$ is defined as follows\cite{prox1}:
\begin{align}
\label{definitionOfProx}
\mathrm{prox}_{\gamma \mathcal{F}}(\mathbf{x}):=\underset{\mathbf{y}\in\mathcal{H}} {\operatorname{argmin}} \,\mathcal{F}(\mathbf{y})+\dfrac{1}{2\gamma}\|\mathbf{y}-\mathbf{x}\|_2^2.
\end{align}
Let us consider the convex conjugate function of $\mathcal{F}$, defined as follows: 
\begin{align}
\mathcal{F}^{*}(\mathbf{x}):=\underset{\mathbf{\mathbf{y}}\in \mathcal{H}}{\operatorname{sup}} \,\ip{\mathbf{x}}{\mathbf{y}}-\mathcal{F}(\mathbf{y}).
\end{align}
According to the Moreau's identity\cite[Theorem 3.1 (ii)]{moreausIdentitiy}, $\mathrm{prox}_{\gamma \mathcal{F}^{*}}$ can be computed via $\mathrm{prox}_{\mathcal{F}/\gamma}$ as
\begin{align}
\mathrm{prox}_{\gamma \mathcal{F}^{*}}(\mathbf{x})=\mathbf{x}-\gamma\mathrm{prox}_{\mathcal{F}/\gamma}(\mathbf{x}/\gamma).
\end{align}

The proximity operator of $\iota_C$ equals to the metric projection onto $C$, i.e.,
\begin{align}
\label{proxOfIndicationFunction}
\mathrm{prox}_{\iota_C}(\mathbf{x})=P_C(\mathbf{x}):=\underset{\mathbf{y}\in C}{\operatorname{argmin}} \,\|\mathbf{y}-\mathbf{x}\|_2.
\end{align}
In particular, it can be calculated for $C:=B_{2,\,\varepsilon}^{\mathbf{v}}$ as follows:
\begin{align}
\mathrm{prox}_{\iota_{B_{2,\,\varepsilon}^{\mathbf{v}}}}(\mathbf{x})=
\begin{cases}
\mathbf{x},&\mathrm{if}\, \mathbf{x}\in C,\\
\mathbf{v}+\dfrac{\varepsilon(\mathbf{x}-\mathbf{v})}{\|\mathbf{x}-\mathbf{v}\|_2},&\mathrm{otherwise}.
\end{cases}
\end{align}
The proximity operator of the GKL divergence can be computed as follows:
\begin{align}
\label{proxOfGKL}
\left[\mathrm{prox}_{\gamma\mathrm{GKL}_{\mathbf{v}}}(\mathbf{x})\right]_i=\dfrac{1}{2}([\mathbf{x}]_i-\gamma\eta+\sqrt{([\mathbf{x}]_i-\gamma\eta)^2+4\gamma [\mathbf{v}]_i}).
\end{align}

\subsection{Primal-Dual Splitting Algorithm}

\label{PDShead}
Let us review a PDS algorithm~\cite{PDS2},\reply{\cite{PDS3}}, which is a basis of our PnP strategy. Let $\mathcal{X}$ and $\mathcal{Y}$ be two real Hilbert spaces. Suppose that $f\in\Gamma_0(\mathcal{X})$ is differentiable with its gradient $\nabla f$ being $\beta$-Lipschitz continuous with $\beta>0$, $g\in\Gamma_0(\mathcal{X})$ and $h\in\Gamma_0(\mathcal{Y})$ are proximable functions, i.e., the proximity operators have a closed form solution or can be computed efficiently, and $\myOp{L}:\mathcal{X}\to\mathcal{Y}$ is a bounded linear operator. Consider the following convex optimization problem:
\begin{align}
\label{PDSprimal}
\underset{\mathbf{x}\in\mathcal{X}}{\operatorname{minimize}} \,f(\mathbf{x})+g(\mathbf{x})+h(\myOp{L}\mathbf{x}),
\end{align}
where the set of solutions to~\eqref{PDSprimal} is assumed to be nonempty. Its dual formulation is given as follows \cite[Chap. 15]{ConvexBook}:
\begin{align}
&\label{PDSdual}\underset{\mathbf{y}\in\mathcal{Y}}{\operatorname{minimize}} \,(f+g)^{*}(-\myOp{L}^{*}\mathbf{y})+h^{*}(\mathbf{y}).
\end{align}

The PDS algorithm jointly solves both \eqref{PDSprimal} and \eqref{PDSdual} by the following iterative procedure:
\begin{empheq}[left=\empheqlfloor]{alignat=2}
\begin{split}
\label{PDSupdate}
&\mathbf{x}_{n+1}=\mathrm{prox}_{\gamma_1 g}\left(\mathbf{x}_{n}-\gamma_1(\nabla f(\mathbf{x}_{n})+\myOp{L}^{*}\mathbf{y}_{n})\right),\\
&\mathbf{y}_{n+1}=\mathrm{prox}_{\gamma_2 h^{*}}\left(\mathbf{y}_{n}+\gamma_2 \myOp{L}(2\mathbf{x}_{n+1}-\mathbf{x}_{n})\right),
\end{split}
\end{empheq}
where $\gamma_1>0$ and $\gamma_2>0$ are \reply{parameters} that satisfy 
\begin{align}
    \label{PDSinequality}
    \gamma_1\left(\dfrac{\beta}{2}+\gamma_2\|\myOp{L}\|_{\mathrm{op}}^2\right)<1.
\end{align}

Additionally, we consider the following monotone inclusion problem derived from the optimality conditions of \eqref{PDSprimal} and \eqref{PDSdual}:
\begin{align}
\label{PDSSolutionSet}
\begin{split}
\mathrm{Find}\,&(\hat{\mathbf{x}},\,\hat{\mathbf{y}})\in\,\mathcal{X}\times\mathcal{Y}\\
\mathrm{s.t.}&
\begin{pmatrix}\mathbf{0}\\\mathbf{0}\end{pmatrix}\in
\begin{pmatrix}
\partial g(\hat{\mathbf{x}})+\myOp{L}^{*}\hat{\mathbf{y}}+\nabla f(\hat{\mathbf{x}})\\
-\myOp{L}\hat{\mathbf{x}}+\partial h^{*}(\hat{\mathbf{y}})
\end{pmatrix},
\end{split}
\end{align}
where $\partial(\cdot)$ denotes the subdifferential of a convex function. Under a certain mild condition, the solutions to \eqref{PDSSolutionSet} coincide with the solutions to \eqref{PDSprimal} and \eqref{PDSdual} (see~\cite{PDS2} and \reply{\cite{PDS3}} for details).

\subsection{Firmly Nonexpansive Denoiser\cite{Pesquet2021}}
\label{firmlyNonexpansiveDenoiser}
We review the method proposed in~\cite{Pesquet2021} for training a firmly nonexpansive denoiser represented by DnCNN, which is essential for the construction of convergent PnP algorithms. Basically, they add a specific penalty during training to ensure the firm nonexpansiveness of the denoiser.

Let $J:\mathbb{R}^{K}\to\mathbb{R}^K$ be a Gaussian denoiser represented by DnCNN. To train a firmly nonexpansive denoiser, we first consider the operator $Q:=2J-\mathrm{Id}$, which is assumed to be differentiable.

By the definition of firm nonexpansiveness, $J$ is firmly nonexpansive if and only if $Q$ is nonexpansive. Thus, our objective is to ensure the nonexpansiveness of $Q$ rather than the firm nonexpansiveness of $J$. If $Q$ is differentiable, its nonexpansiveness is equivalent to its Jacobian $\nabla Q$ satisfying the following condition for all $\mathbf{x}\in\mathbb{R}^{K}$:
\begin{align}
    \label{NonexpansivenessAndJacobian}
    \|\nabla Q(\mathbf{x})\|_\mathrm{sp}\leq1,
\end{align}
where $\|\cdot\|_\mathrm{sp}$ represents the spectral norm.

Based on the above fact, the authors of \cite{Pesquet2021} proposed a penalty to encourage $Q$ to satisfy the condition in \eqref{NonexpansivenessAndJacobian}, and defined a new loss function expressed by the sum of the squared prediction error and the penalty. Specifically, the loss function is calculated for each image as
\begin{align}
    \label{PenaltyOfDNCNN}
    \|J\left(\mathbf{x}_\ell)-\overline{\mathbf{x}}_\ell\|^2_2+\tau\,\mathrm{max}(\|\nabla Q(\mathbf{\tilde{x}}_\ell)\|_\mathrm{sp}^2,\,1-\xi\right).
\end{align}
Here, $\tau>0$ is a penalization parameter, $\xi>0$ is a parameter that defines a safety margin, $\overline{\mathbf{x}}_\ell$ is the $\ell$-th ground truth image in the training dataset, $\mathbf{x}_\ell$ is the image corrupted with Gaussian noise, and $\tilde{\mathbf{x}}_\ell$ is the image calculated by the following equation: 
\begin{align}
\label{Qselection}
\tilde{\mathbf{x}}_\ell=\rho\overline{\mathbf{x}}_\ell+(1-\rho)J(\mathbf{x}_\ell),
\end{align}
where $\rho$ is a random value chosen according to a uniform distribution in the range from $0$ to $1$. 

Note that $\|\nabla Q(\mathbf{\tilde{x}}_\ell)\|_\mathrm{sp}$ can be computed using the power iterative method\cite{PowerIteration}. For more detailed information on the training method, see \cite[Section 3.2]{Pesquet2021}.
 
\section{Proposed Method}
\label{sec:proposedmethod}
We propose a general convergent PnP-PDS, demonstrate its application to image restoration with specific cases of our PnP-PDS, and provide detailed discussions on its differences from related methods. For clarity, we present a flowchart in Fig.~\ref{flowchart} that illustrates the relationships between our theoretical results.

\begin{figure}[t]
    \centering
    \includegraphics[keepaspectratio, width=\linewidth]{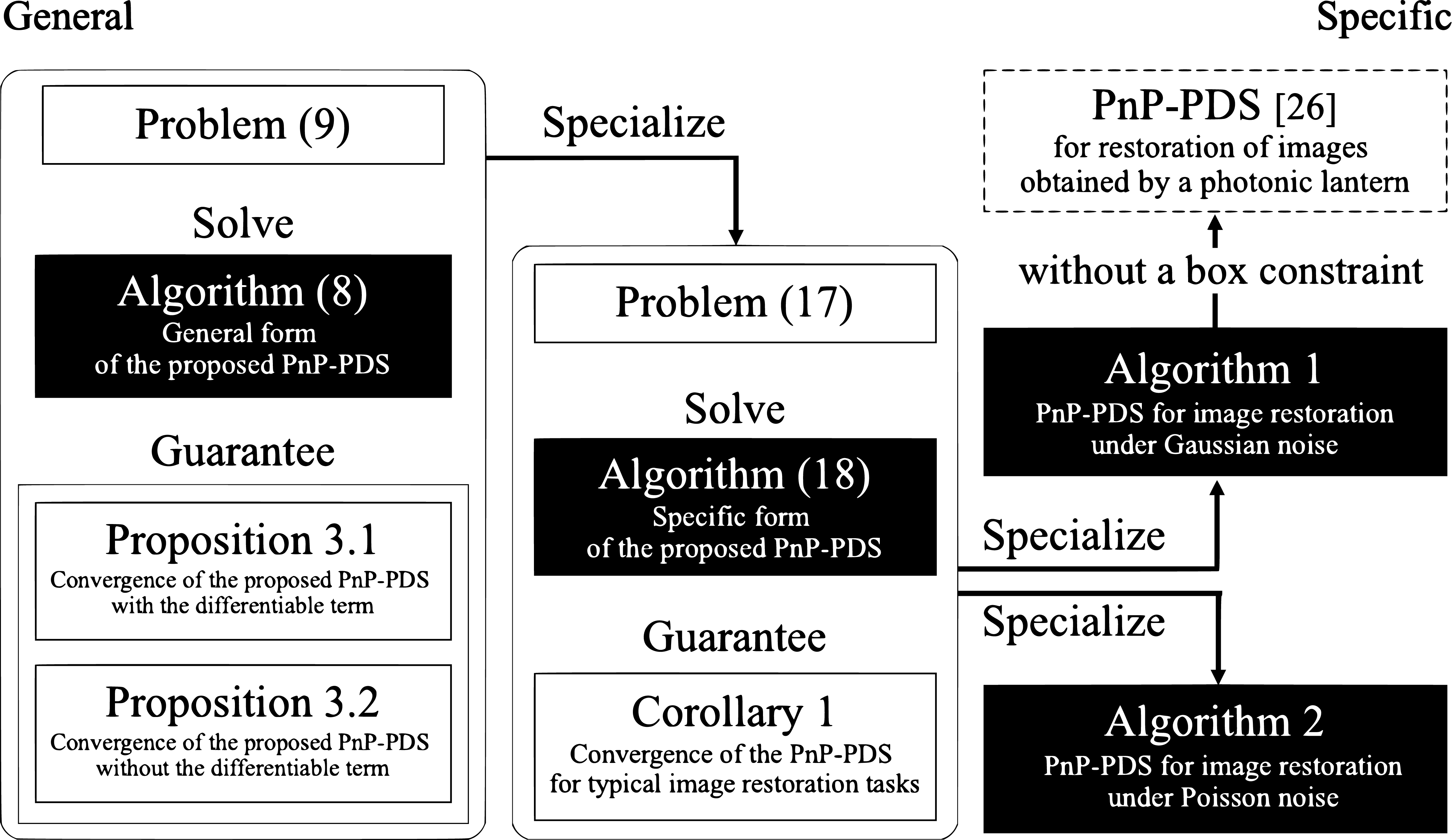}
    \caption{A flowchart showing the basic structure of the theoretical discussions in this paper. The dotted and solid lines represent existing and new results, respectively.}
    \label{flowchart}
\end{figure}
\subsection{General Form of Convergent PnP-PDS and Its Convergence Analysis}
We propose a convergent PnP-PDS that substitutes the proximity operator of $g$ in~\eqref{PDSupdate} with a firmly nonexpansive denoiser $J$ represented by DnCNN. The iterative procedures of our PnP-PDS is given by
    \begin{empheq}[left=\empheqlfloor]{alignat=2}
    \label{PnP-PDSIteration}
    \begin{split}
                &\widetilde{\mathbf{x}}_{n+1}=J\left(\mathbf{x}_n-\gamma_1(\nabla f(\mathbf{x}_n)+\mathbf{L}^{*}\mathbf{y}_{n})\right),\\
    &\widetilde{\mathbf{y}}_{n+1}=\mathrm{prox}_{\gamma_2 h^{*}}\left(\mathbf{y}_{n}+\gamma_2 \mathbf{L}(2\widetilde{\mathbf{x}}_{n+1}-\mathbf{x}_{n})\right),\\
        &\mathbf{x}_{n+1}=\rho_n\widetilde{\mathbf{x}}_{n+1}+(1-\rho_n)\mathbf{x}_{n},\\
        &\mathbf{y}_{n+1}=\rho_n\widetilde{\mathbf{y}}_{n+1}+(1-\rho_n)\mathbf{y}_{n},
    \end{split}
\end{empheq}
where $\gamma_1,\,\gamma_2$ and $\{\rho_n\}$ are the \reply{parameters} and the definitions of $f$, $h$, and $\mathbf{L}$ are the same as in~\eqref{PDSprimal}. \replySec{The parameters $\gamma_1$ and $\gamma_2$ are fixed constants across the iterations.}

The following two propositions present our main theoretical results on the convergence of PnP-PDS in~\eqref{PnP-PDSIteration}. Proposition~\ref{propOfConvergence} corresponds to the case with the differentiable term $f$, while Proposition~\ref{propOfConvergenceWhenBetaIsZero} corresponds to the case without it.

\begin{proposition}
\label{propOfConvergence}
Suppose that $\beta>0$. Let $J:\mathcal{X}\to\mathcal{X}$ be a firmly nonexpansive denoiser and defined over the full domain. Then, there exists a \textit{maximally monotone} operator $A_J$ such that $J=(A_J+\mathrm{Id})^{-1}$. Assume that the solution set of the following monotone inclusion problem:
\begin{align}
\label{PnP-PDSSolutionSet}
\begin{split}
\mathrm{Find}\,&(\hat{\mathbf{x}},\,\hat{\mathbf{y}})\in\,\mathcal{X}\times\mathcal{Y}\\
&\mathrm{s.t.}
\begin{pmatrix}\mathbf{0}\\\mathbf{0}\end{pmatrix}\in
\begin{pmatrix}
\gamma_1^{-1}A_J(\hat{\mathbf{x}})+\mathbf{L}^{*}\hat{\mathbf{y}}+\nabla f(\hat{\mathbf{x}})\\
-\mathbf{L}\hat{\mathbf{x}}+\partial h^{*}(\hat{\mathbf{y}})
\end{pmatrix}
\end{split}
\end{align}
is nonempty, and suppose that the following conditions hold:
\begin{enumerate}[(i)]
\item $\dfrac{1}{\gamma_1}-\gamma_2\|\mathbf{L}\|_{\mathrm{op}}^2>\dfrac{\beta}{2}$,
\item $\rho_n\in\left(0,\,\delta\right)$,
\item $\sum_{n\in\mathbb{N}}\rho_n(\delta-\rho_n)=\infty$,
\end{enumerate}
where $\gamma_1>0$, $\gamma_2>0$, and $\delta$ is defined as follows:
\begin{align}
\delta:=2-\dfrac{\beta}{2}\left(\dfrac{1}{\gamma_1}-\gamma_2\|\mathbf{L}\|_{\mathrm{op}}^{2}\right)^{-1}.
\end{align} 
Then, the sequence $\{\mathbf{x}_n,\,\mathbf{y}_n\}_{n\in\mathbb{N}}$ generated by~\eqref{PnP-PDSIteration} weakly converges to a solution to~\eqref{PnP-PDSSolutionSet}.
\end{proposition}

\begin{proposition}
\label{propOfConvergenceWhenBetaIsZero}
    Suppose that $f=0$. Let $J:\mathcal{X}\to\mathcal{X}$ be a firmly nonexpansive denoiser and defined over the full domain. Then, there exists a \textit{maximally monotone} operator $A_J$ such that $J=(A_J+\mathrm{Id})^{-1}$. Assume that the solution set of~\eqref{PnP-PDSSolutionSet} is nonempty, and the following conditions hold:
    \begin{enumerate}[(i)]
    \item $\dfrac{1}{\gamma_1}-\gamma_2\|\mathbf{L}\|_{\mathrm{op}}^2>0$,
    \item $\rho_n\in\left(0,\,2\right)$,
    \item $\sum_{n\in\mathbb{N}}\rho_n(2-\rho_n)=\infty$,
    \end{enumerate}
    where $\gamma_1>0$ and $\gamma_2>0$.
Then, the sequence $\{\mathbf{x}_n,\,\mathbf{y}_n\}_{n\in\mathbb{N}}$ generated by~\eqref{PnP-PDSIteration} weakly converges to a solution to~\eqref{PnP-PDSSolutionSet}.
\end{proposition}

\begin{proof}[Proof Sketch of Propositions 3.1 and 3.2] 
\reply{Recall that the set-valued operator $A_J=J^{-1}-\mathrm{Id}$ is maximally monotone if and only if $J$ is firmly nonexpansive and has the full domain~\cite[Proposition 23.8 (iii)]{ConvexBook}. Hence, $\gamma_1^{-1}A_J$ is also maximally monotone \cite[Proposition 20.22]{ConvexBook}. Therefore, Proposition~\ref{propOfConvergence} and Proposition~\ref{propOfConvergenceWhenBetaIsZero} are justified by~\cite[Theorem 3.1]{PDS2} and~\cite[Theorem 3.2]{PDS2}, respectively. These proofs build upon and directly apply the standard convergence analysis of the PDS algorithm. For completeness and ease of reference, we provide the detailed proofs of these propositions in the Appendix.}
\end{proof}

\begin{rem}[Connection to the convex optimization problem in~\eqref{PDSprimal}]
\label{remConnectionToMinimization}
The monotone inclusion problem in~\eqref{PnP-PDSSolutionSet} has the same form as~\eqref{PDSSolutionSet} with $\partial g=\gamma_1^{-1}A_J$. The classical monotone operator theory states that there exists $g\in\Gamma_0(\mathbb{R}^N)$ such that $\partial g=\gamma_1^{-1}A_J$ if and only if $\gamma_1^{-1}A_J$ is a \textit{maximal cyclically monotone operator}, a specific type of maximally monotone operator~\cite[Theorem 22.18]{ConvexBook}. However, the operator $\gamma_{1}^{-1}A_J$, where $J=(A_J+\mathrm{Id})^{-1}$, is not necessarily maximal cyclically monotone only because $J$ is firmly nonexpansive. Thus, the existence of $g\in\Gamma_0(\mathbb{R}^N)$ cannot always be guaranteed. Furthermore, it should be noted that $A_J$ is scaled by $\gamma_1^{-1}$ in the solution set of the monotone inclusion problem~\eqref{PnP-PDSSolutionSet}. This implies that the value of $\gamma_1$ affects the solution set, unlike the standard PDS algorithm. 
\end{rem}
% CondatのTheorem3.3と同じように、beta=0でもパラメータの不等号を等号つき不等号にできると補足する？
%\begin{rem}[Connection to the convex optimization problem in \eqref{PDSprimal}]
%\label{remConnectionToMinimization}
%\end{rem}
\subsection{Applications to Image Restoration}
\label{ApplicationToImageRestoration}
Now we move on to the applications of our general PnP-PDS to image restoration. First, we derive a special case of our PnP algorithm in~\eqref{PnP-PDSIteration}, which is set up for typical image restoration problems. Then, we demonstrate its application to specific image restoration tasks with two types of noise contamination: Gaussian and Poisson. In the following, $\mathbf{u}\in\mathbb{R}^{K}$ and $\mathbf{v}\in\mathbb{R}^{K}$ denote original and observed images, respectively. Let $\mathbf{\Phi}:\mathbb{R}^{K}\to\mathbb{R}^{K}$ be an operator representing some observation process, which may include blurring or random sampling. We assume that the pixel values of each image are normalized to the range of 0 to 1.

First, we introduce an observation model that covers many image restoration scenarios as follows:
\begin{align}
    \label{observationModelOfTypiclaIR}
    \mathbf{v}=\mathcal{N}(\mathbf{\Phi} \mathbf{u}),
\end{align}
where $\mathcal{N}:\mathbb{R}^K\to\mathbb{R}^K$ represents noise contamination (not necessarily additive). Image restoration under~\eqref{observationModelOfTypiclaIR} can be formulated as the following optimization problem:
\begin{align}
    \label{PnP-PDSCorollaryMinimization}
    \begin{split}
        \underset{\mathbf{\hat{u}}\in\mathbb{R}^K}{\operatorname{minimize}} \,\mathcal{D}(\mathbf{\Phi}\mathbf{\hat{u}})+\mathcal{R}_J(\mathbf{\hat{u}})\quad\mathrm{s.t.}\,\mathbf{\hat{u}}\in C,
    \end{split}
\end{align}
or equivalently, 
\begin{align}
    \label{PnP-PDSCorollaryMinimizationIndicator}
    \begin{split}
        \underset{\mathbf{\hat{u}}\in\mathbb{R}^K}{\operatorname{minimize}} \,\mathcal{D}(\mathbf{\Phi}\mathbf{\hat{u}})+\mathcal{R}_J(\mathbf{\hat{u}})+\iota_C(\mathbf{\hat{u}}),
    \end{split}
\end{align}
where $\mathcal{D}\in\Gamma_0(\mathbb{R}^K)$ is a proximable data-fidelity term derived from the noise model $\mathcal{N}$, $\mathcal{R}_J$ is a proximable regularization term, $C$ is a nonempty closed convex set on $\mathbb{R}^K$ to which $\hat{\mathbf{v}}$ is assumed to belong, and $\iota_C$ is the indicator function of $C$. \reply{We also assume that $\mathcal{R}_J$ models the prior distribution associated with the Gaussian denoiser $J$ with $\mathrm{prox}_{\mathcal{R}_J}=J$, since $\mathrm{prox}_{\mathcal{R}_J}$ corresponds to the MAP estimate under additive white Gaussian noise~\cite{PnP0}. Note that $\mathcal{R}_J$ is introduced for convenience as an explicit regularization term, but the existence of such a function is not guaranteed (see Remark~\ref{remConnectionToMinimization}).}

Then, we can derive a convergent PnP-PDS algorithm for this type of image restoration as a special case of our general PnP-PDS in~\eqref{PnP-PDSIteration}. The algorithm and its convergence property are summarized in the following corollary.
\begin{cor}
\label{PnP-PDSCorollary}
Let $J:\mathbb{R}^K\to\mathbb{R}^K$ be a firmly nonexpansive denoiser and defined over the full domain. Then, there exists a \textit{maximally monotone} operator $A_J$ such that $J=(A_J+\mathrm{Id})^{-1}$. Assume that the solution set of the following monotone inclusion problem:
    \begin{align}
    \label{PnP-PDSCorollarySolutionSet}
    \begin{split}
    \mathrm{Find}\,&(\hat{\mathbf{x}},\,\hat{\mathbf{y}}^{(1)},\,\hat{\mathbf{y}}^{(2)})\in\,\mathbb{R}^{3K}\\
    &\mathrm{s.t.}
    \begin{pmatrix}\mathbf{0}\\\mathbf{0}\\\mathbf{0}\end{pmatrix}\in
    \begin{pmatrix}
    \gamma_1^{-1}A_J(\hat{\mathbf{x}})+\mathbf{\Phi}^{*}\hat{\mathbf{y}}^{(1)} + \hat{\mathbf{y}}^{(2)}\\
    -\mathbf{\Phi}\hat{\mathbf{x}}+\partial \mathcal{D}^{*}\left(\hat{\mathbf{y}}^{(1)}\right)\\
    -\hat{\mathbf{x}} + \partial \iota^{*}_C\left(\hat{\mathbf{y}}^{(2)}\right)
    \end{pmatrix}
    \end{split}
    \end{align}
    is nonempty, and suppose that the following conditions hold for $\gamma_1>0$ and $\gamma_2>0$:
    \begin{align}
    \dfrac{1}{\gamma_1}-\gamma_2\left(\|\mathbf{\Phi}\|_{\mathrm{op}}^2+1\right)>0,
    \end{align}
and $\{\rho_n\}$ satisfies the condition in (ii) and (iii) of Proposition~\ref{propOfConvergenceWhenBetaIsZero}.
Then, the sequence $\{\mathbf{x}_n,\,\mathbf{y}^{(1)}_n,\,\mathbf{y}^{(2)}_n\}_{n\in\mathbb{N}}$ generated by the following algorithm:
    \begin{empheq}[left=\empheqlfloor]{alignat=2}
    \label{PnP-PDSCorollaryIteration}
    \begin{split}
    &\widetilde{\mathbf{x}}_{n+1}=J\left(\mathbf{x}_n-\gamma_1\left(\mathbf{\Phi}^{*}\mathbf{y}_{n}^{(1)}+\mathbf{y}_{n}^{(2)}\right)\right),\\
    &\widetilde{\mathbf{y}}_{n+1}^{(1)}=\mathrm{prox}_{\gamma_2 \mathcal{D}^{*}}\left(\mathbf{y}_{n}^{(1)}+\gamma_2 \mathbf{\Phi}(2\widetilde{\mathbf{x}}_{n+1}-\mathbf{x}_{n})\right),\\
    &\widetilde{\mathbf{y}}_{n+1}^{(2)}=\mathrm{prox}_{\gamma_2 \iota_C^{*}}\left(\mathbf{y}_{n}^{(2)}+\gamma_2 (2\widetilde{\mathbf{x}}_{n+1}-\mathbf{x}_{n})\right),\\
    &\mathbf{x}_{n+1}=\rho_n\widetilde{\mathbf{x}}_{n+1}+(1-\rho_n)\mathbf{x}_{n},\\
    &\mathbf{y}_{n+1}^{(1)}=\rho_n\widetilde{\mathbf{y}}_{n+1}^{(1)}+(1-\rho_n)\mathbf{y}_{n}^{(1)},\\
    &\mathbf{y}_{n+1}^{(2)}=\rho_n\widetilde{\mathbf{y}}_{n+1}^{(2)}+(1-\rho_n)\mathbf{y}_{n}^{(2)},
    \end{split}
\end{empheq}
converges to a solution to~\eqref{PnP-PDSCorollarySolutionSet}.
\end{cor}

\reply{This is a direct consequence of Proposition~\ref{propOfConvergenceWhenBetaIsZero}.} Let us define $f(\mathbf{x}),\,h\left(\mathbf{y}^{(1)},\,\mathbf{y}^{(2)}\right)$ in~\eqref{PnP-PDSSolutionSet} as follows:
\begin{align}
f\left(\mathbf{x}\right)&:=0,\\
h(\mathbf{y}^{(1)},\,\mathbf{y}^{(2)})&:=\mathcal{D}(\mathbf{y}^{(1)})+\iota_C(\mathbf{y}^{(2)}),\\
\mathbf{L}:&=(\mathbf{\Phi},\,\mathbf{I}),
\end{align}
where $\mathbf{I}$ represents an identity matrix. Then, the monotone inclusion problem in~\eqref{PnP-PDSSolutionSet} is reduced to \eqref{PnP-PDSCorollarySolutionSet} \cite[Proposition 13.30, Proposition 16.9, Corollary 16.48]{ConvexBook}, and we obtain Algorithm \eqref{PnP-PDSCorollaryIteration} from \eqref{PnP-PDSIteration}. \reply{The condition regarding $\gamma_1$ and $\gamma_2$ is obtained from the assumption in (i) of Proposition~\ref{propOfConvergenceWhenBetaIsZero} and the following fact~\cite[Theorem 1]{blockOperatorNorm}:}
\begin{align}
\|\mathbf{L}\|_{\mathrm{op}}^2\leq \|\mathbf{\Phi}\|_{\mathrm{op}}^2+\|\mathbf{I}\|_{\mathrm{op}}^2=\|\mathbf{\Phi}\|_{\mathrm{op}}^2+1.
\end{align}

\subsubsection{Image Restoration under Gaussian Noise}
\label{GaussianApplication}
Consider the observation model in~\eqref{observationModelOfTypiclaIR} with
$\mathcal{N}(\mathbf{x}):=\mathbf{x}+\mathbf{n}$, where $\mathbf{n}\in\mathbb{R}^{K}$ is additive white Gaussian noise. Image restoration under this observation model can be formulated as the following optimization problem:
\begin{align}
\label{GaussianPDSFormulation}
\underset{\mathbf{\hat{u}}\in\mathbb{R}^K}{\operatorname{minimize}} \,\mathcal{R}_J(\mathbf{\hat{u}}) \quad\mathrm{s.t.}\,
\begin{cases}
    \|\mathbf{\Phi} \mathbf{\hat{u}}-\mathbf{v}\|_2\leq\varepsilon,\\
    \mathbf{\hat{u}}\in\left[0,\,1\right]^{K},
\end{cases}
\end{align}
where $\varepsilon>0$ is the radius of the $\ell_2$ data-fidelity constraint, which can be determined by the standard deviation of $\mathbf{n}$.

To apply the algorithm in~\eqref{PnP-PDSCorollaryIteration}, we define $\mathcal{D}(\cdot)$ and $C$ in~\eqref{PnP-PDSCorollaryMinimization} as follows:
\begin{align}
\mathcal{D}(\cdot)&:=\iota_{B_{2,\,\varepsilon}^{\mathbf{v}}}(\cdot),\\
C&:=\left[0,\,1\right]^{K},
\end{align}
where $B_{2,\,\varepsilon}^{\mathbf{v}}$ is a $\mathbf{v}$-centered $\ell_2$-norm ball with the radius $\varepsilon$. Then, the problem in~\eqref{GaussianPDSFormulation} can be expressed as the form of~\eqref{PnP-PDSCorollaryMinimization}, allowing the application of~\eqref{PnP-PDSCorollaryIteration}. If we set $\rho_n = 1$ for all $n \in \mathbb{N}$, we obtain Algorithm~\ref{GaussianPDSAlgorithm}.

\begin{rem}[Constrained formulation of data-fidelity]
In~\eqref{GaussianPDSFormulation}, data-fidelity is expressed as a hard constraint using the $\ell_2$-norm ball. The benefits of expressing data-fidelity or regularization terms as hard constraints have been illustrated in the literature~\cite{ConstraintPDS1, ConstraintPDS2, constraint0, constraint4, l0proj}. We will discuss such advantages in the specific case of~\eqref{GaussianPDSFormulation} for PnP methods in Section~\ref{ComparisonWithPnPFBS}.
\end{rem}

\begin{algorithm}[t]
  \caption{Image restoration under Gaussian noise by PnP-PDS.}
  \label{GaussianPDSAlgorithm}
  \begin{algorithmic}[1]
    \AlgInput $\mathbf{u}_{0},\,\mathbf{w}_{0}^{(1)},\,\mathbf{w}_{0}^{(2)},\,\gamma_1>0,\,\gamma_2>0,\,\varepsilon>0$
    \While {a stopping criterion is not satisfied}
      \State $\mathbf{u}_{n+1} \leftarrow J\left(\mathbf{u}_{n}-\gamma_1\left(\mathbf{\Phi}^{*} \mathbf{w}_{n}^{(1)} + \mathbf{w}_{n}^{(2)}\right)\right)$;
      \State $\overline{\mathbf{w}}_{n}^{(1)} \leftarrow \mathbf{w}_{n}^{(1)}+\gamma_2\mathbf{\Phi}(2\mathbf{u}_{n+1}-\mathbf{u}_{n})$;
    \State $\overline{\mathbf{w}}_{n}^{(2)} \leftarrow \mathbf{w}_{n}^{(2)}+\gamma_2(2\mathbf{u}_{n+1}-\mathbf{u}_{n})$;
      \State $\mathbf{w}_{n+1}^{(1)} \leftarrow \overline{\mathbf{w}}_n^{(1)} - \gamma_2 P_{B^{\mathbf{v}}_{2,\,\varepsilon}}(\frac{1}{\gamma_2}\overline{\mathbf{w}}_n^{(1)})$;
      \State $\mathbf{w}_{n+1}^{(2)} \leftarrow \overline{\mathbf{w}}_n^{(2)} - \gamma_2 P_{[0,\,1]^K}(\frac{1}{\gamma_2}\overline{\mathbf{w}}_n^{(2)})$;
    \EndWhile
    \AlgOutput $\mathbf{u}_{n}$
  \end{algorithmic}
\end{algorithm}

\subsubsection{Image Restoration under Poisson Noise}
\label{PoissonApplication}
Image restoration under Poisson noise has been studied across various domains, including medical imaging, astronomical imaging, and remote sensing. The observation process considering Poisson noise contamination can be modeled by~\eqref{observationModelOfTypiclaIR} when $\mathcal{N}(\mathbf{x}):=\mathcal{P}_\eta(\mathbf{x})$, where $\mathcal{P}_\eta$ represents the corruption by Poisson noise with the scaling coefficient $\eta$.

In this observation model, the data-fidelity is expressed as the GKL divergence (refer to \cite{2007CombettsPoisson}), and image restoration can be reduced to the following optimization problem:
\begin{align}
\label{PoissonPDSFormulation}
\underset{\mathbf{\hat{u}}\in\mathbb{R}^K}{\operatorname{minimize}} \,\lambda\mathrm{GKL}_{\mathbf{v}}(\mathbf{\Phi}\mathbf{\hat{u}})+\mathcal{R}_J(\mathbf{\hat{u}})
\quad\mathrm{s.t.}\,
    \mathbf{\hat{u}}\in\left[0,\,1\right]^{K}
\end{align}
where $\lambda>0$ is a parameter that determines the balance between data-fidelity and regularization (see~\eqref{definitionOfGKL} for the definition of $\mathrm{GKL}$).

Let us define $\mathcal{D}(\cdot)$ and $C$ in~\eqref{PnP-PDSCorollaryMinimization} as follows:
\begin{align}
\mathcal{D}(\cdot)&:=\lambda\mathrm{GKL}_{\mathbf{v}}(\cdot),\\
C&:=\left[0,\,1\right]^{K}.
\end{align}
Then, the problem in~\eqref{PoissonPDSFormulation} is reformulated as~\eqref{PnP-PDSCorollaryMinimization}, and the iterative procedure in~\eqref{PnP-PDSCorollaryIteration} becomes applicable. Setting $\rho_n = 1$ for all $n \in \mathbb{N}$, the resulting algorithm is given as Algorithm~\ref{PoissonPDSAlgorithm}.

\begin{algorithm}[t]
  \caption{Image restoration under Poisson noise by PnP-PDS.}
  \label{PoissonPDSAlgorithm}
  \begin{algorithmic}[1]
    \AlgInput $\mathbf{u}_{0},\,\mathbf{w}_{0}^{(1)},\,\mathbf{w}_{0}^{(2)},\,\gamma_1>0,\,\gamma_2>0,\,\lambda>0$
    \While {a stopping criterion is not satisfied}
      \State $\mathbf{u}_{n+1} \leftarrow J\left(\mathbf{u}_{n}-\gamma_1\left(\mathbf{\Phi}^{*} \mathbf{w}_{n}^{(1)}+\mathbf{w}_{n}^{(2)}\right)\right)$;
      \State $\overline{\mathbf{w}}_{n}^{(1)} \leftarrow \mathbf{w}_{n}^{(1)}+\gamma_2\mathbf{\Phi}(2\mathbf{u}_{n+1}-\mathbf{u}_{n})$;
      \State $\overline{\mathbf{w}}_{n}^{(2)} \leftarrow \mathbf{w}_{n}^{(2)}+\gamma_2(2\mathbf{u}_{n+1}-\mathbf{u}_{n})$;
      \State $\mathbf{w}_{n+1}^{(1)} \leftarrow \overline{\mathbf{w}}_n^{(1)} - \gamma_2 \mathrm{prox}_{\frac{\lambda}{{\gamma_2}} \mathrm{GKL}}(\frac{1}{\gamma_2}\overline{\mathbf{w}}_n^{(1)})$;
    \State $\mathbf{w}_{n+1}^{(2)} \leftarrow \overline{\mathbf{w}}_n^{(2)} - \gamma_2 {P_{[0,\,1]^K}}(\frac{1}{\gamma_2}\overline{\mathbf{w}}_n^{(2)})$;
    \EndWhile
    \AlgOutput $\mathbf{u}_{n}$
  \end{algorithmic}
\end{algorithm}

\subsection{Comparison with Prior Works}
\label{ComparisonWithOtherMethods}
\subsubsection{Convergent PnP-FBS \cite{Pesquet2021}}
\label{ComparisonWithPnPFBS}
Pesquet et al. proposed a PnP-FBS algorithm using DnCNN as a convergent PnP approach~\cite{Pesquet2021}. This method can be applied to image restoration under Gaussian noise, addressing the following optimization problem:
\begin{gather}
\label{GaussianFBSFormulation}
\underset{\mathbf{\hat{u}}\in\mathbb{R}^K}{\operatorname{minimize}} \, \frac{\lambda}{2}\|\mathbf{\Phi\hat{u}}-\mathbf{v}\|_2^2+\mathcal{R}_J(\mathbf{\hat{u}}),
\end{gather}
where the first term is an $\ell_2$ data-fidelity term, and $\lambda>0$ is a balancing parameter. 

This convergent PnP-FBS is a landmark method that first ensures the convergence of PnP-FBS with the characterization of the solution set under realistic assumptions. At the same time, it should be noted that the problem in~\eqref{GaussianPDSFormulation}, where the data-fidelity term is represented as a hard constraint, cannot be solved by PnP-FBS. 

\reply{Consequently, the proposed PnP-PDS~(Algorithm~\ref{GaussianPDSAlgorithm}) inherits the versatility of a once-trained denoiser in this setting, compared to PnP-FBS. This advantage arises from its high flexibility in parameter selection. In the case of the constrained formulation in~\eqref{GaussianPDSFormulation}, a natural choice for $\varepsilon$ is given as follows~\cite{ConstraintPDS2}:
\begin{align}
\label{idealParameterForPnPPDSGaussian}
\varepsilon = \sigma \sqrt{K},
\end{align}
where $\sigma$ denotes the standard deviation of the Gaussian noise in the observed images. We remark that the convergence of the proposed PnP-PDS is guaranteed regardless of the choice of $\varepsilon$.}

\reply{In contrast, PnP-FBS limits the flexibility of the denoiser in certain cases, due to constraints on the parameter to ensure convergence. Assume that the denoiser is trained with a noise level of $\sigma_J$, and that $\mathbf{\Phi}$ represents a circular convolution with kernel $h$. In the additive formulation in~\eqref{GaussianFBSFormulation}, the authors suggest the following value for $\lambda$~\cite{Pesquet2021}:
\begin{align}
\label{HeuristicsForLambdaInPnPFBS}
\lambda_{\mathrm{opt}} := \frac{\sigma_J}{2 \sigma \|h\|}.
\end{align}
However, $\lambda$ must also satisfy the condition to ensure the convergence of PnP-FBS:
\begin{align}
\label{ConditionForLambdaInPnPFBS}
\lambda < \frac{2}{\|\boldsymbol{\Phi}\|_{\mathrm{op}}^2}.
\end{align}
Hence, if $\lambda_{\mathrm{opt}}$ does not satisfy the constraint in~\eqref{ConditionForLambdaInPnPFBS}, the convergence of PnP-FBS cannot be ensured with its optimal parameter. This limitation was observed in our experiments (see Section~\ref{ComparisonOfParameterRobustness}).}

\reply{Another limitation of PnP-FBS is its relatively narrow applicability to problem formulations. Specifically, it does not handle additional constraints explicitly, such as the box constraint in~\eqref{GaussianPDSFormulation}, and cannot be applied to the Poisson noise setting in~\eqref{PoissonPDSFormulation}, which is formulated with the generalized Kullback–Leibler (GKL) divergence.}

\subsubsection{Convergent PnP-ADMM~\cite{ConvergentPnPADMM, IstaAndADMMPnPConvergence}}
\label{ComparisonWithPnPADMM}
Several works have proposed PnP algorithms based on ADMM (PnP-ADMM) and established their convergence guarantees. Nair et al. have proposed a convergent PnP-ADMM and provided extensive analysis on its fixed-point~\cite{IstaAndADMMPnPConvergence}. Sun et al. have proposed a convergent PnP-ADMM algorithm that handles multiple data-fidelity terms, making it practical for large scale datasets~\cite{ConvergentPnPADMM}. In both studies, the convergence properties are ensured under the firm nonexpansiveness of the denoiser. PnP-ADMM can address a wider range of tasks than PnP-FBS, including image restoration with Poisson noise.

The main advantage of the proposed PnP-PDS over PnP-ADMM is the elimination of inner iterations. To solve a subproblem in PnP-ADMM algorithms, it is often necessary to perform inner iterations, which results in increasing computational cost and unstable numerical convergence. On the other hand, the proposed PnP-PDS does not need to perform them thanks to the fact that the computations associated with the linear operator $L$ are fully decoupled from the proximity operator of $h^*$ in the algorithm.

\subsubsection{Convergent PnP-PDS\cite{PnPPDSByPesquet}}
Garcia et al. proposed a convergent PDS-based PnP algorithm for restoration of images obtained by a photonic lantern~\cite{PnPPDSByPesquet}, published contemporaneously with our conference paper~\cite{PnPPDSOurs}. They consider the observation process $\mathbf{v}=\mathbf{\Phi}\mathbf{u}+\mathbf{n}$,
where $\mathbf{v}$ is an observed image, $\mathbf{u}$ is an original image, $\mathbf{\Phi}$ is a linear degradation process, and $\mathbf{n}\in\mathbb{R}^{K}$ is an additive noise with an unknown distribution, which is assumed to satisfy $\|\mathbf{n}\|_2<\varepsilon$ for $\varepsilon>0$. Based on this observation process, they address the following optimization problem by a convergent PnP-PDS:
\begin{align}
\label{ProblemOfPnPPDSByPesquet}
\underset{\mathbf{\hat{u}}\in\mathbb{R}^K}{\operatorname{minimize}} \,\mathcal{R}_{J}(\mathbf{\hat{u}}) \quad\mathrm{s.t.}\,
    \|\mathbf{\Phi} \mathbf{\hat{u}}-\mathbf{v}\|_2\leq\varepsilon.
\end{align}

We show the relationship between this algorithm and our PnP-PDS algorithm in Fig.~\ref{flowchart}. In essence, the algorithm in~\cite{PnPPDSByPesquet} is a special case of our Algorithm~\ref{GaussianPDSAlgorithm} without the box constraint.

\reply{We note that introducing a box constraint ensures that the solution to the optimization (or monotone inclusion) problem remains within a certain convex set, which in turn promotes more stable behavior of PnP iterations for some experimental settings. The experimental results in Section~\ref{sec:experiments} illustrate this effect.}

\subsubsection{Regularization by Denoising (RED)\cite{RED,REDAnalysis}}
\label{ComparisonWithRED}
Romano et al. proposed a framework called RED, which defines an explicit objective function by incorporating the output of a denoiser. For image restoration under Gaussian noise and image restoration under Poisson noise, the optimization problems with RED are given respectively as
\begin{align}
\label{GaussianREDFormulation}
&\underset{\mathbf{\hat{u}}\in\mathbb{R}^K}{\operatorname{minimize}} \,\lambda\|\mathbf{\Phi} \mathbf{\hat{u}}-\mathbf{v}\|^2_2+\mathbf{\hat{u}}^\top(\mathbf{\hat{u}}-\mathfrak{G}(\mathbf{\hat{u}})),\\
\label{PoissonREDFormulation}
&\underset{\mathbf{\hat{u}}\in\mathbb{R}^K}{\operatorname{minimize}} \,\lambda\mathrm{GKL}_\mathbf{v}(\hat{\mathbf{u}})+\mathbf{\hat{u}}^\top(\mathbf{\hat{u}}-\mathfrak{G}(\mathbf{\hat{u}})),
\end{align}
\reply{where $\mathfrak{G}$ represents a Gaussian denoiser. In RED, the optimization algorithms are constructed by computing the gradient of the second terms in~\eqref{GaussianREDFormulation} and~\eqref{PoissonREDFormulation} using the following expression:}

\reply{
\begin{align}
    \label{REDapprox}
    \nabla(\mathbf{\hat{u}}^\top(\mathbf{\hat{u}}-\mathfrak{G}(\mathbf{\hat{u}})))=
    2(\hat{\mathbf{u}}-\mathfrak{G}(\mathbf{\hat{u}})).
\end{align}
For differentiable $\mathfrak{G}$, this equation holds if and only if $\mathfrak{G}$ is locally homogeneous and has a symmetric Jacobian~\cite{RED, REDAnalysis}. For non-differentiable $\mathfrak{G}$, a more recent study~\cite{REDPRO} demonstrates that~\eqref{REDapprox} still holds under the assumption that $\mathrm{Id} - \mathfrak{G}$ is \textit{maximally cyclically monotone}, or that $\mathfrak{G}$ is \textit{cyclically firmly nonexpansive}.  Furthermore, if $\mathfrak{G}$ is also nonexpansive in addition to either of these conditions, we can establish a convergent algorithm based on gradient computations, such as gradient descent or FBS.} 

The main advantage of using RED is the simplicity resulting from the existence of a clear convex objective function. In contrast, PnP algorithms including our PnP-PDS do not necessarily have a convex objective function. Instead, they solve a monotone inclusion problem expressed as~\eqref{PnP-PDSSolutionSet}.

However, RED also has a considerable difficulty: the assumptions for the denoiser. First, local homogeneity is an unrealistic property, especially for deep denoisers. In fact, the gradients of the regularization terms computed in the RED algorithms often contain substantial computational errors due to the lack of local homogeneity in the denoisers~\cite{REDAnalysis}. \reply{Second, \textit{cyclically firm nonexpansiveness~\cite{REDPRO}} is even stronger assumption than the firm nonexpansiveness assumed in our work. These challenges make it difficult to ensure the convergence of RED algorithms when we employ state-of-the-art denoisers, including DnCNN.}

\section{Experiments}
\label{sec:experiments}

\begin{table}[t]
    \centering
    \resizebox{\linewidth}{!}{
    \begin{tabular}{lll}
        \toprule
        \textbf{Method} & \textbf{Algorithm} & \textbf{Regularization} \\
        \midrule
        PnP-FBS & FBS & FNE\textsuperscript{*} DnCNN \\
        PnP-ADMM & ADMM & FNE\textsuperscript{*} DnCNN \\
        PnP-PDS (Unstable) & PDS & DnCNN (not FNE\textsuperscript{*}) \\
        PnP-PDS (\replySec{NoBox}) & PDS & FNE\textsuperscript{*} DnCNN \\
        TV & PDS & TV \\
        RED (Gaussian) & SD\textsuperscript{†} & FNE\textsuperscript{*} DnCNN \\
        RED (Poisson) & PDS & FNE\textsuperscript{*} DnCNN \\
        DRUNet & HQS\textsuperscript{‡} & DRUNet \\
        BPG & BPG & Bregman Score Denoiser \\
        Proposed & PDS & FNE\textsuperscript{*} DnCNN \\
        \bottomrule
        \multicolumn{3}{l}{\footnotesize{\textsuperscript{*}FNE: firmly nonexpansive \qquad \textsuperscript{†}SD: steepest descent}}\vspace{-1mm}\\ 
        \multicolumn{3}{l}{\footnotesize{\textsuperscript{‡}HQS: half-quadratic splitting}}
    \end{tabular}
    }
    \caption{Summary of Compared Methods, Their Optimization Algorithms, and Regularization.}
    \label{tab:method-summary}
\end{table}
We performed two types of experiments to demonstrate the stability, performance, and versatility of the proposed method. The first experiment involves deblurring/inpainting under Gaussian noise, and the second experiment involves deblurring/inpainting under Poisson noise. The purpose of these experiments is to confirm the following two facts:
\begin{itemize}
\setlength{\leftskip}{-10pt}
    \item The proposed method operates stably and converges in all experimental settings.
    \item In comparison to other state-of-the-art methods, the proposed method demonstrates higher restoration performance thanks to its stability.
 \end{itemize} 
 
\subsection {Experimental Setup}

\def\kernelsize{0.08}
\begin{figure*}[t]
    \centering
    \captionsetup[subfigure]{justification=centering}

    \begin{minipage}[t]{\textwidth}
        \centering
        \begin{subfigure}[b]{\kernelsize\textwidth}
            \includegraphics[width=\linewidth]{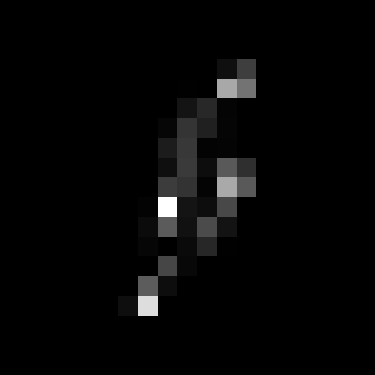}
            \caption*{(a)\\$0.2246$}
        \end{subfigure}
        \begin{subfigure}[b]{\kernelsize\textwidth}
            \includegraphics[width=\linewidth]{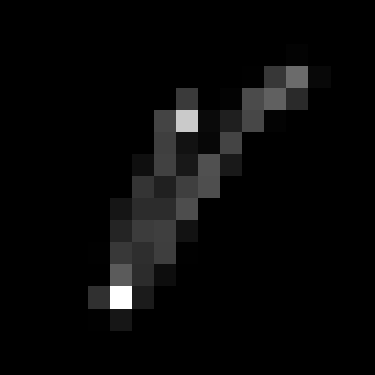}
            \caption*{(b)\\$0.1933$}
        \end{subfigure}
        \begin{subfigure}[b]{\kernelsize\textwidth}
            \includegraphics[width=\linewidth]{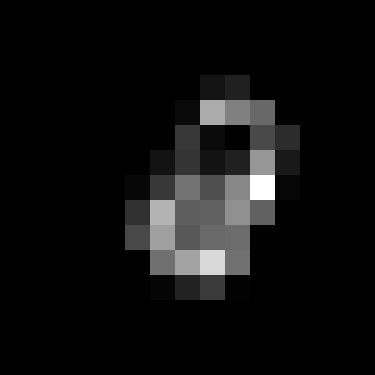}
            \caption*{(c)\\$0.1907$}
        \end{subfigure}
        \begin{subfigure}[b]{\kernelsize\textwidth}
            \includegraphics[width=\linewidth]{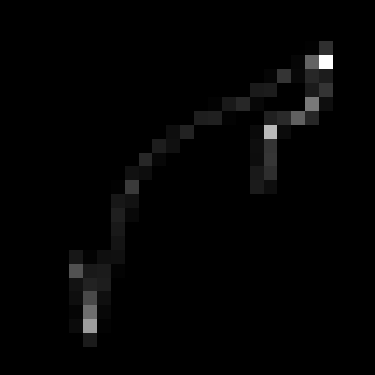}
            \caption*{(d)\\$0.1778$}
        \end{subfigure}
        \begin{subfigure}[b]{\kernelsize\textwidth}
            \includegraphics[width=\linewidth]{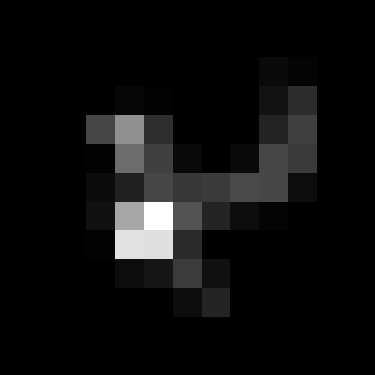}
            \caption*{(e)\\$0.2255$}
        \end{subfigure}
        \begin{subfigure}[b]{\kernelsize\textwidth}
            \includegraphics[width=\linewidth]{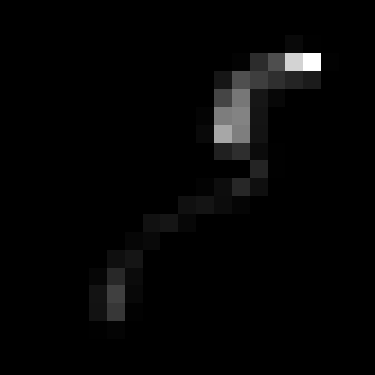}
            \caption*{(f)\\$0.2163$}
        \end{subfigure}
        \begin{subfigure}[b]{\kernelsize\textwidth}
            \includegraphics[width=\linewidth]{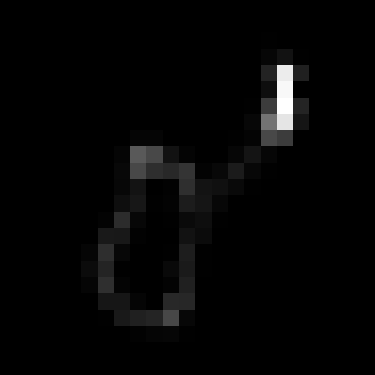}
            \caption*{(g)\\$0.1917$}
        \end{subfigure}
        \begin{subfigure}[b]{\kernelsize\textwidth}
            \includegraphics[width=\linewidth]{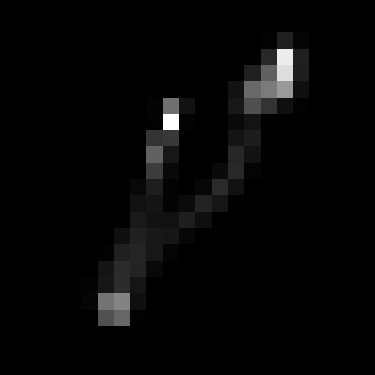}
            \caption*{(h)\\$0.1737$}
        \end{subfigure}
        \begin{subfigure}[b]{\kernelsize\textwidth}
            \includegraphics[width=\linewidth]{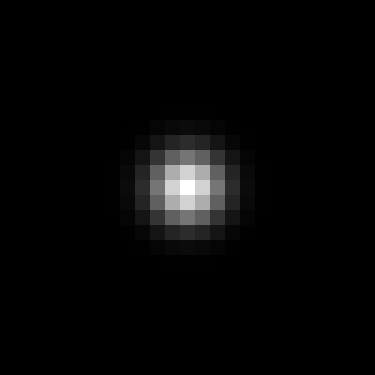}
            \caption*{(i)\\$0.1763$}
        \end{subfigure}
        \begin{subfigure}[b]{\kernelsize\textwidth}
            \includegraphics[width=\linewidth]{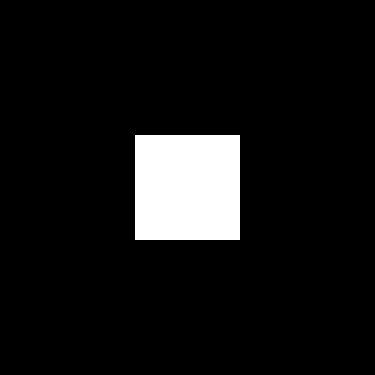}
            \caption*{(j)\\$0.1429$}
        \end{subfigure}        
    \end{minipage}

    \caption{\replySec{Blur kernels used in our experiments. (a)–(h) correspond to kernels (1)–(8) in\cite{kernel}, while (i) and (j) correspond to the ``GaussianA'' and ``Square'' setup in~\cite{Bertocchi_2020}, respectively. The Frobenius norm of each kernel is shown in the second row.}}
    \label{kernel_visualize}
\end{figure*}

In the two experiments, we considered two types of observation operators: blur and random sampling. 
\replySec{For the blur case, we used circular convolution with the ten kernels (a)–(j) from~\cite{kernel,Bertocchi_2020}, visualized in Fig.~\ref{kernel_visualize}, and we scaled them so that \(\|\mathbf{\Phi}\|_{\mathrm{op}}=1\). In Sec.~IV-B-1 we evaluated all ten kernels (a)–(j), whereas in the other sections we used kernel~(a) as a representative case.
For the random sampling case, $20\%$ of the image pixels were randomly selected and masked. The operator norm also satisfies $\|\mathbf{\Phi}\|_{\mathrm{op}}=1$ in this case.}

We evaluated the restoration performance by Peak Signal to Noise Ratio (PSNR) and Structural Similarity Index Measure (SSIM). To investigate stability, the update rate $c_n$ for $n\in\mathbb{N}\backslash\{0\}$ was defined by $c_n := {\|\mathbf{u}_n-\mathbf{u}_{n-1}\|_2} / {\|\mathbf{u}_{n-1}\|_2},$ where $\mathbf{u}_n$ is the intermediate solution at the $n$-th iteration by each algorithm.

We used color images for Gaussian noise case and grayscale images for Poisson noise case. The color images were randomly sampled from ImageNet~\cite{ImageNet}, consisting of $7$ images cropped to $128\times128$ pixels. The grayscale images were selected from 3 widely used image sets \cite{Pesquet2021}, all with a size of $256\times256$ pixels.

\newcommand\bestfbs[1]{\mathbf{#1}}
\begin{table*}[t]
\centering
\setlength{\tabcolsep}{8pt}
\resizebox{\linewidth}{!}{%
\replySec{
\begin{tabular}{lcccccccccccc} \toprule
Noise level  & Method & (a) & (b) & (c) & (d) & (e) & (f) & (g) & (h) & (i) & (j) & Average \\ \midrule
\multirow{2}{*}{$\sigma=0.0025$} 
 & PnP-FBS & \replyThird{$31.28^*$} & \replyThird{$30.42^*$} & \replyThird{$30.28^*$} & \replyThird{$30.11^*$} & \replyThird{$31.76^*$} & \replyThird{$- -$} & \replyThird{$30.14^*$}  & \replyThird{$30.02^*$}  & \replyThird{$28.15^*$}  & \replyThird{$27.75^*$}  & $- -$ \\
 & Proposed       & $\bestfbs{37.27}$ & $\bestfbs{36.63}$ & $\bestfbs{35.68}$ & $\bestfbs{36.64}$ & $\bestfbs{36.69}$ & $\bestfbs{37.49}$ & $\bestfbs{35.43}$ & $\bestfbs{35.58}$ & $\bestfbs{29.70}$ & $\bestfbs{30.85}$ & $\bestfbs{35.20}$ \\ \midrule
\multirow{2}{*}{$\sigma=0.005$} 
 &PnP-FBS & \replyThird{$31.16^*$}  & \replyThird{$30.35^*$}  & \replyThird{$30.22^*$}  & \replyThird{$30.04^*$}  & \replyThird{$31.67^*$}  & $- -$ & \replyThird{$30.07^*$}  & \replyThird{$29.95^*$}  & \replyThird{$28.12^*$}  & \replyThird{$27.71^*$}  & $- -$ \\
 & Proposed       & $\bestfbs{34.04}$ & $\bestfbs{33.46}$ & $\bestfbs{32.95}$ & $\bestfbs{33.37}$ & $\bestfbs{34.01}$ & $\bestfbs{34.46}$ & $\bestfbs{32.71}$ & $\bestfbs{32.70}$ & $\bestfbs{28.87}$ & $\bestfbs{29.33}$ & $\bestfbs{32.59}$\\ \midrule
\multirow{2}{*}{$\sigma=0.01$}  
 & PnP-FBS & $30.58$ & \replyThird{$29.98^*$}  & \replyThird{$29.95^*$}  & \replyThird{$29.69^*$}  & $31.24$ & $31.26$ & \replyThird{$29.76^*$}  & \replyThird{$29.60^*$}  & \replyThird{$28.00^*$}  & \replyThird{$27.56^*$}  & \replyThird{$29.76$} \\
 & Proposed       & $\bestfbs{30.97}$ & $\bestfbs{30.52}$ & $\bestfbs{30.50}$ & $\bestfbs{30.37}$ & $\bestfbs{31.55}$ & $\bestfbs{31.74}$ & $\bestfbs{30.20}$ & $\bestfbs{30.09}$ & $\bestfbs{28.16}$ & $\bestfbs{27.85}$ & $\bestfbs{30.20}$ \\ \midrule
\multirow{2}{*}{$\sigma=0.02$}  
 & PnP-FBS & $\bestfbs{28.26}$ & $\bestfbs{27.92}$ & $\bestfbs{28.35}$ & $\bestfbs{27.66}$ & $29.34$ & $29.27$ & $\bestfbs{28.18}$ & $\bestfbs{27.94}$ & $\bestfbs{27.43}$ & $\bestfbs{26.80}$ & $\bestfbs{28.12}$\\
 & Proposed       & $28.15$ & $27.70$ & $28.14$ & $27.41$ & $\bestfbs{29.36}$ & $\bestfbs{29.28}$ & $28.15$ & $27.73$ & $27.07$ & $26.39$ & $27.94$\\ \midrule
\multirow{2}{*}{$\sigma=0.04$}  
 &PnP-FBS & $25.55$ & $24.96$ & $25.99$ & $24.34$ & $26.98$ & $26.43$ & $25.79$ & $25.09$ & $\bestfbs{26.41}$ & $25.18$ & $25.67$\\
 & Proposed       & $\bestfbs{26.00}$ & $\bestfbs{25.73}$ & $\bestfbs{26.29}$ & $\bestfbs{25.58}$ & $\bestfbs{27.05}$ & $\bestfbs{26.80}$ & $\bestfbs{26.13}$ & $\bestfbs{25.94}$ & $26.18$ & $\bestfbs{25.42}$ & $\bestfbs{26.11}$\\ \bottomrule
\multicolumn{13}{l}{\small{$^{*}$ \replySec{\replyThird{As the algorithm did not converge for $\lambda$ values in the range $0.8\lambda_\mathrm{opt}$ to $1.25\lambda_\mathrm{opt}$, we present the results for $\lambda=1.99$.}}}}

\end{tabular}
}
}
\caption{\replySec{PSNR [dB] Values for 10 Blur Kernels Across Different Noise Levels. The Symbol ``$- -$'' Denotes That the Algorithm Diverged and PSNR Values Could not be Obtained. The Better Results are Highlighted in Bold.}}
\label{tab:PSNR_blurKernels}
\end{table*}

In each experiment, we selected the appropriate methods from PnP-FBS \cite{Pesquet2021}, PnP-ADMM \cite{ConvergentPnPADMM}, PnP-PDS (Unstable) \cite{PnPPDS}, PnP-PDS (\replySec{NoBox})~\cite{PnPPDSByPesquet}, TV \cite{constraint0}, RED \cite{RED}, \reply{DRUNet~\cite{DRUNet}}%, \replyThird{DRUNet (8 iter)~\cite{DRUNet}}
\reply{, and Bregman Proximal Gradient (BPG)~\cite{BregmanProximalGradient}}, and compared them with the proposed method. \reply{The optimization algorithms and regularization for each method are summarized in Table~\ref{tab:method-summary}. Details of PnP-FBS, PnP-ADMM, and RED are given in Section \ref{ComparisonWithOtherMethods}. \replySec{PnP-PDS (NoBox) corresponds to the proposed method without the box constraint, with all other settings aligned. We adopted the stopping criterion proposed by the authors for BPG. For the other methods, a fixed iteration count was used. Specifically, in the Poisson noise case with $\eta\in\{1,\,2,\,10\}$, the number of iterations was fixed to $4800$ for deblurring and $12000$ for inpainting. In the other cases, the number of iterations was fixed at $1200$ for the blur case and $3000$ for the random sampling case.}} \replySec{For DRUNet, we employed the half-quadratic splitting algorithm, and used the same fixed iteration counts as the other methods.} \replyThird{In addition, we also evaluate DRUNet (8 iter), where DRUNet is executed with 8 iterations of the HQS algorithm, to examine its performance in the low-iteration regime commonly used in practice.}

\replySec{For all settings, we used the DnCNN~\cite{Pesquet2021} trained with the noise level $\sigma_J$, for the regularization in PnP-FBS, PnP-ADMM, PnP-PDS (\replySec{NoBox}), RED, and the proposed method. This DnCNN is learned as described in Section~\ref{firmlyNonexpansiveDenoiser} to promote its firm nonexpansiveness. For a fair comparison with PnP-FBS, whose the optimal parameter depends on $\sigma_J$, we adopted the model distributed by Pesquet et al., trained with $\sigma_J=0.01$ in Section~\ref{GaussianExperiment}, and used the same setting in Section~\ref{PoissonExperiment}. In Section~\ref{PriorExperiment}, we investigate the influence of $\sigma_J$ by examining two cases: (i) $\sigma_J$ fixed to one of $\{0.0075,\,0.01,\,0.05,\,0.1\}$ and (ii) $\sigma_J$ randomly sampled from a uniform distribution over $[0,\,0.1]$ for each training sample.}

Throughout our experiments, the \reply{parameters} $\gamma_1$ and $\gamma_2$ in Algorithm~\ref{GaussianPDSAlgorithm} and~Algorithm~\ref{PoissonPDSAlgorithm} were set to $0.5$ and $0.99$, respectively, satisfying the condition in Proposition~\ref{propOfConvergenceWhenBetaIsZero}. All of our experiments were conducted on a Windows 11 PC, equipped with an Intel Core i9 3.70GHz processor and 32.0GB of RAM. 

\subsection {Deblurring/Inpainting under Gaussian Noise}
\label{GaussianExperiment}
\reply{
We performed experiments on deblurring/inpainting under Gaussian noise. Here, we investigate the versatility of a once-trained denoiser in the proposed method, and compare its performance with existing methods.
}

\subsubsection{\reply{Versatility of a Once-Trained Denoiser}}
\label{ComparisonOfParameterRobustness}
\replySec{We investigate the versatility of the once-trained denoiser with the fixed value of $\sigma_J$ in different algorithms. Since the theoretical convergence guarantees in this work are built upon PnP-FBS, let us confirm that the versatility is enhanced by extending PnP-FBS to PnP-PDS.}

\reply{As discussed in Section~\ref{ComparisonWithPnPFBS}, the proposed method handles $\ell_2$ data-fidelity constraint, whereas PnP-FBS handles $\ell_2$ data-fidelity term. Following~\eqref{idealParameterForPnPPDSGaussian}, we set $\varepsilon$ in~\eqref{GaussianPDSFormulation} as} \replySec{$\varepsilon = \alpha \varepsilon_\mathrm{opt}$ where $\varepsilon_\mathrm{opt}:=\sigma \sqrt{K}$ and evaluated the performance over $\alpha \in [0.8,\, 1.25]$. For PnP-FBS, we followed the heuristic in~\eqref{HeuristicsForLambdaInPnPFBS} and set $\lambda = \zeta \lambda_{\mathrm{opt}}$ where $\lambda_\mathrm{opt}:=\sigma_J/(2\sigma\|h\|)$ and $\|h\|$ is the Frobenius norm of the blur kernel $h$. We evaluated the performance over $\zeta \in [0.8,\, 1.25]$.} \replyThird{When the algorithm diverged within this range, we instead used $\lambda = 1.99$.}

\replySec{Table~\ref{tab:PSNR_blurKernels} shows the variation in restoration performance across different noise levels and blur kernels. At low noise levels and for kernels with small Frobenius norms, PnP-FBS diverges with $\lambda$ close to $\lambda_\mathrm{opt}$ as discussed in Section~\ref{ComparisonWithPnPFBS}. For example, in the case of $\sigma = 0.01$ and kernel~(b), we obtain $\lambda_{\mathrm{opt}} \approx 2.587$ because $\sigma_J = 0.01$ and $2\sigma\|h\| \approx 0.0039$. However, the convergence condition~\eqref{ConditionForLambdaInPnPFBS} requires $\lambda < 2$\replyThird{, forcing the use of suboptimal parameter choices}. This discrepancy stems from the dependence of $\lambda_\mathrm{opt}$ on $\sigma_J$.}

\replySec{In contrast, the proposed method exhibits reasonable performance across all noise levels. This consistency demonstrates the high versatility of the given denoiser, thanks to the decoupling of $\varepsilon_\mathrm{opt}$ and $\sigma_J$ (see also the discussion in Section~\ref{ComparisonWithPnPFBS}). Moreover, we note that $\varepsilon_\mathrm{opt}$ is independent of the blur kernels, which enables stable performance across all kernels under a consistent parameter choice.}

\subsubsection{\reply{Restoration Performance}}
\begin{table}[t]
\centering
\setlength{\tabcolsep}{9pt}
\resizebox{\linewidth}{!}{%
\begin{tabular}{lcccccc} \toprule
    \reply{Noise level $\sigma$} & \reply{$0.0025$} & \reply{$0.005$} & \reply{$0.01$} & \reply{$0.02$} & \reply{$0.04$}  \\ \midrule
   \reply{Deblurring} & \reply{$0.82$} & \reply{$0.86$} & \reply{$0.92$} & \reply{$0.96$} & \reply{$1.00$}  \\
   \reply{Inpainting} & \reply{$0.90$} & \reply{$0.82$} & \reply{$0.82$} & \reply{$1.00$} & \reply{$1.00$}\\
   \bottomrule
   
\end{tabular}
\caption{\reply{The Values of $\alpha$ Used in Each Setting for the Proposed Method.}}
\label{GaussianParameterTable}
}
\end{table}

\begin{table*}[t]
\centering
\setlength{\tabcolsep}{8pt}
\resizebox{\textwidth}{!}{%
\begin{tabular}{lcccccccccc} \toprule
    & \multicolumn{5}{c}{Deblurring} & \multicolumn{5}{c}{Inpainting} \\ \cmidrule(lr){2-6}\cmidrule(lr){7-11}
   Noise level $\sigma$ & $0.0025$ & $0.005$ & $0.01$ & $0.02$ & $0.04$ & $0.0025$ & $0.005$ & $0.01$ & $0.02$ & $0.04$ \\ \midrule
   PnP-FBS~\cite{Pesquet2021} & \reply{$31.28$} & \reply{$31.16$} & \reply{$30.59$} & \reply{$28.19$} & \reply{$23.95$} & \reply{$- -$} & \reply{$- -$} & \reply{$- -$} & \reply{$- -$} & \reply{$- -$} \\
   PnP-PDS (Unstable)~\cite{PnPPDS} & $23.79$ & $23.67$ & $23.78$ & $24.60$ & $23.35$ & $33.65$ & $33.36$ & $32.79$ & $31.51$ & $29.58$ \\
   PnP-PDS (\replySec{NoBox})~\cite{PnPPDSByPesquet} & $\underline{37.23}$ & $34.03$ & $31.00$ & $28.24$ & $\mathbf{26.00}$ & $- -$ & $- -$ & $- -$ & $- -$ & $- -$ \\
   TV~\cite{constraint0} & $31.97$ & $30.40$ & $28.78$ & $27.10$ & $25.38$ & $\underline{34.19}$ & $\underline{33.90}$ & $\underline{33.22}$ & $\underline{31.90}$ & $\underline{29.85}$ \\
   RED~\cite{RED} & $34.11$ & $33.38$ & $30.89$ & $\mathbf{28.30}$ & $24.87$ & $- -$ & $- -$ & $- -$ & $- -$ & $- -$ \\
   \reply{DRUNet~\cite{DRUNet}} & \replySec{$32.60$} & \replySec{$31.80$} & \replySec{$28.54$} & \replySec{$- -$} & \replySec{$- -$} & \replySec{$33.21$} & \replySec{$32.68$} & \replySec{$31.36$} & \replySec{$26.71$} & \replySec{$- -$}\\
   \replyThird{DRUNet (8 iter)~\cite{DRUNet}} & \replyThird{$35.33$} & \replyThird{$\mathbf{34.50}$} & \replyThird{$\mathbf{31.74}$} & \replyThird{$21.65$} & \replyThird{$11.38$} & \replyThird{$13.84$} & \replyThird{$13.84$} & \replyThird{$13.82$} & \replyThird{$13.77$} & \replyThird{$13.57$}\\
   PnP-PDS (Proposed) & $\mathbf{37.24}$ & $\underline{34.04}$ & $\underline{31.01}$ & $\underline{28.26}$ & $\mathbf{26.00}$ & $\mathbf{35.61}$ & $\mathbf{35.25}$ & $\mathbf{34.43}$ & $\mathbf{33.04}$ & $\mathbf{31.00}$ \\
   \bottomrule
  %\multicolumn{11}{l}{\small{$^{*}$ \replySec{When the number of iterations is extremely small, higher results are obtained.}}}
   \end{tabular}
}
\caption{The Average PSNR [dB] Values for Deblurring and Inpainting on the 7 Images From ImageNet. The Best Value is Highlighted in Bold and the Second Best Value is Underlined.}
\label{GaussianResultTablePerformance}
\end{table*}

\begin{figure}[t]
    \centering
    \includegraphics[keepaspectratio, width=\linewidth]{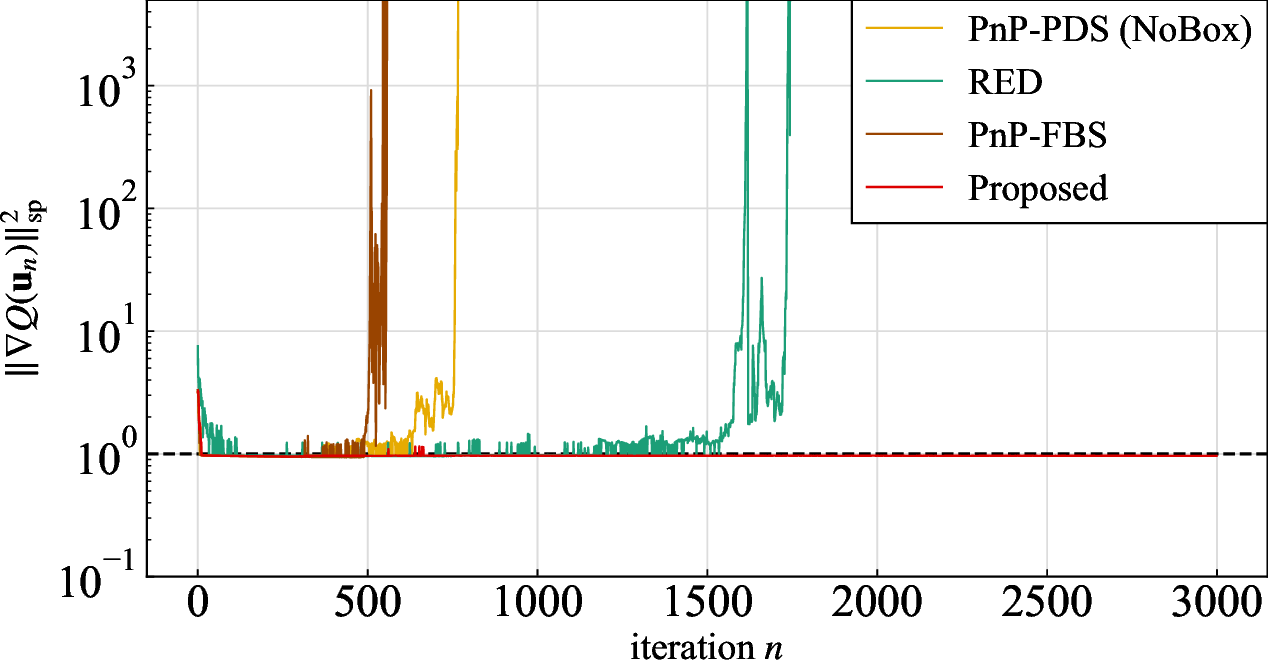}
    \caption{The evolution of $\|\nabla Q(\mathbf{u}_n)\|_{\mathrm{sp}}^2$ of PnP-PDS (NoBox), PnP-FBS, RED, and the proposed method for inpainting at $\sigma=0.01$. The values in each plot are computed as the mean over 7 images.}
    \label{GaussianBoxConstraintStudy}
\end{figure}

\newcommand{\myimgDiv}[4]{%
  \begin{subfigure}[t]{#4\linewidth}
    \centering
    \includegraphics[width=\linewidth]{#1}
    \vspace{-5mm}
    \caption*{\replySec{#3}}%
  \end{subfigure}%
}
\def\divimgsize{0.113}
\begin{figure*}[t]
    \centering
    \captionsetup[subfigure]{justification=centering}
    \begin{minipage}[t]{\linewidth}
        \begin{minipage}[c]{0.02\textwidth}
            \raggedleft
            \rotatebox{90}{}
        \end{minipage}
        \begin{minipage}[c]{0.02\textwidth}
            \raggedleft
            \rotatebox{90}{}
        \end{minipage}
        \replySec{
        \makebox[\divimgsize\textwidth]{$\mathbf{u}_1$}
        \makebox[\divimgsize\textwidth]{$\mathbf{u}_{753}$}
        \makebox[\divimgsize\textwidth]{$\mathbf{u}_{755}$}
        \makebox[\divimgsize\textwidth]{$\mathbf{u}_{757}$}
        \makebox[\divimgsize\textwidth]{$\mathbf{u}_{759}$}
        \makebox[\divimgsize\textwidth]{$\mathbf{u}_{761}$}
        \makebox[\divimgsize\textwidth]{$\mathbf{u}_{1001}$}
        \makebox[\divimgsize\textwidth]{$\mathbf{u}_{3000}$}}
    \end{minipage}
    \begin{minipage}[t]{\linewidth}
        \vspace{-1.8mm}
        \replySec{
            \begin{minipage}[c]{1em}
                \rotatebox{90}{\hspace{5.7em}PnP-PDS}
            \end{minipage}
            \begin{minipage}[c]{1em}
                \rotatebox{90}{\hspace{5.7em}(NoBox)}
            \end{minipage}
        }
        \myimgDiv{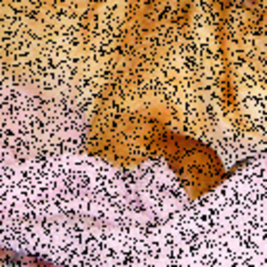}{$\mathbf{u}_1$}{$3.285$}{\divimgsize}
        \myimgDiv{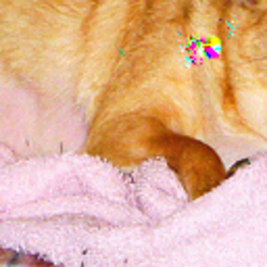}{$\mathbf{u}_{753}$}{$2492.4$}{\divimgsize}
        \myimgDiv{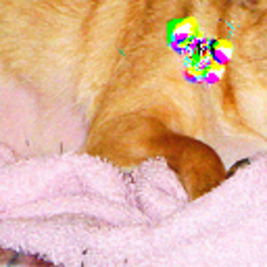}{$\mathbf{u}_{755}$}{$- -$}{\divimgsize}
        \myimgDiv{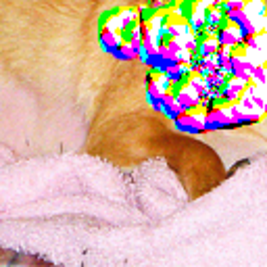}{$\mathbf{u}_{757}$}{$- -$}{\divimgsize}
        \myimgDiv{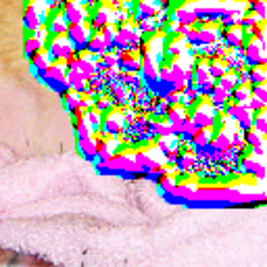}{$\mathbf{u}_{759}$}{$- -$}{\divimgsize}
        \myimgDiv{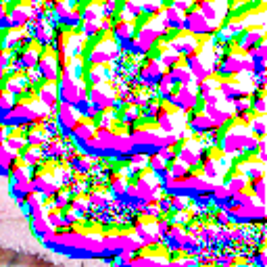}{$\mathbf{u}_{761}$}{$- -$}{\divimgsize}
        \myimgDiv{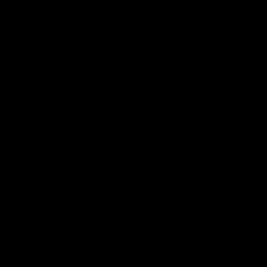}{$\mathbf{u}_{1001}$}{$- -$}{\divimgsize}
        \myimgDiv{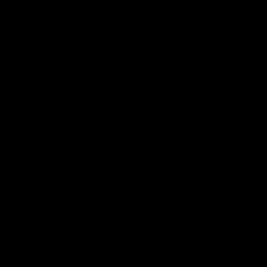}{$\mathbf{u}_{3000}$}{$- -$}{\divimgsize}
    \end{minipage}
    \begin{minipage}[t]{\linewidth}
       \vspace{-2em}
       \replySec{
        \begin{minipage}[c]{1em}
            \rotatebox{90}{}
        \end{minipage}
        \begin{minipage}[c]{1em}
            \rotatebox{90}{\hspace{5.7em}Proposed}
        \end{minipage}
        }
        \myimgDiv{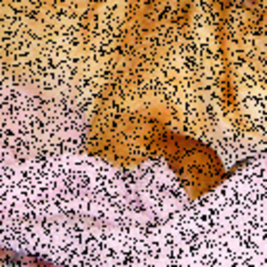}{$\mathbf{u}_1$}{$3.327$}{\divimgsize}
        \myimgDiv{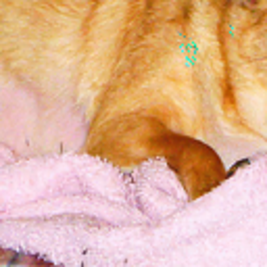}{$\mathbf{u}_{753}$}{$1.043$}{\divimgsize}
        \myimgDiv{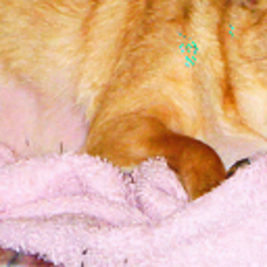}{$\mathbf{u}_{755}$}{$1.063$}{\divimgsize}
        \myimgDiv{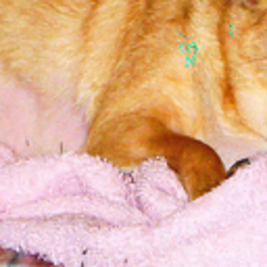}{$\mathbf{u}_{757}$}{$1.039$}{\divimgsize}
        \myimgDiv{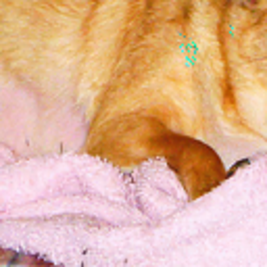}{$\mathbf{u}_{759}$}{$1.068$}{\divimgsize}
        \myimgDiv{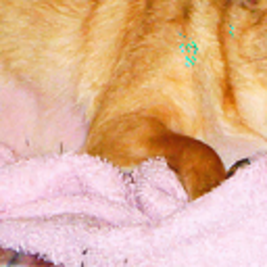}{$\mathbf{u}_{761}$}{$1.036$}{\divimgsize}
        \myimgDiv{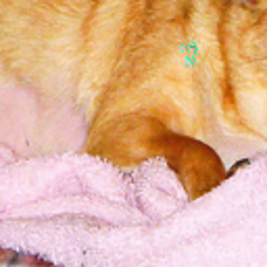}{$\mathbf{u}_{1001}$}{$1.036$}{\divimgsize}
        \myimgDiv{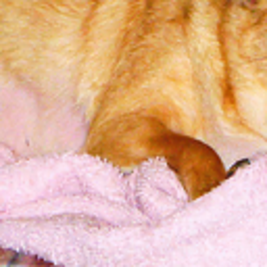}{$\mathbf{u}_{3000}$}{$1.035$}{\divimgsize}
       \vspace{-3em}
    \end{minipage}

    \caption{\replySec{Visualization of $\mathbf{u}_n$ of PnP-PDS (NoBox) and the proposed method, along with the corresponding values of $\|\nabla Q(\mathbf{u}_n)\|_\mathrm{sp}^2$ for each $\mathbf{u}_n$. For visualization, each pixel value is clipped to the range $[0,1]$. The symbol ``$- -$'' indicates that the value of $\|\nabla Q(\mathbf{u}_n)\|_\mathrm{sp}^2$ could not be computed. Since $\mathbf{u}_{1001}$ and $\mathbf{u}_{3000}$ of PnP-PDS (NoBox) contained NaN values in Python, they could not be visualized.}}
    \label{DivergenceVisualize}
\end{figure*}

\begin{figure*}[t]
    \centering
    \includegraphics[keepaspectratio, width=\linewidth]{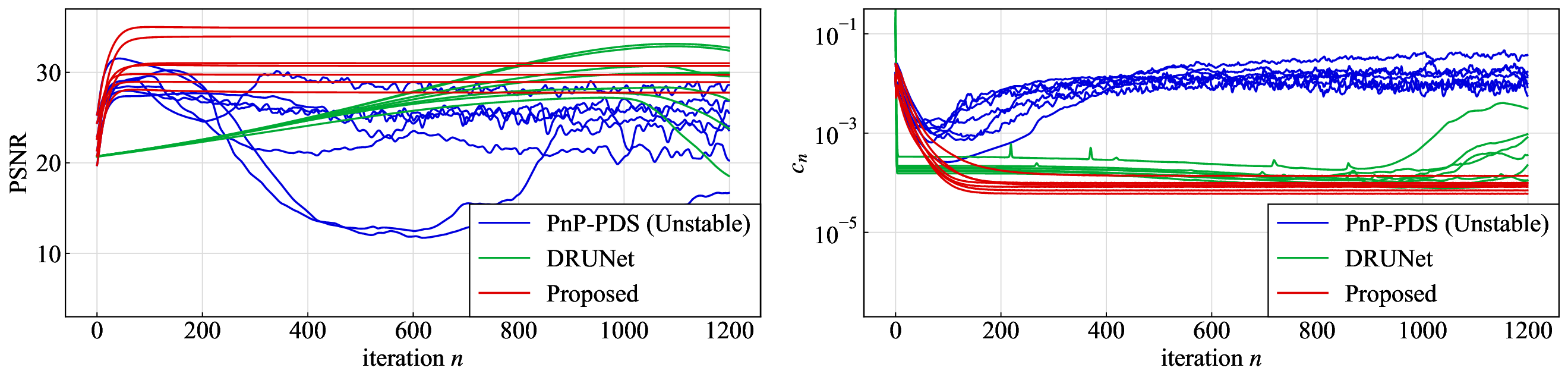}
    \caption{\vspace{-1mm}The evolution of PSNR [dB] and $c_n$ of PnP-PDS (Unstable), \replyThird{DRUNet, and }the proposed method over iterations for deblurring at $\sigma=0.01$. Each plot in the graph corresponds to each image extracted from ImageNet~\cite{ImageNet}.}
    \label{GaussianResultGraph}
\end{figure*}

\newcommand{\myimg}[4]{
    \begin{minipage}{#4\textwidth}
        \centering
        \includegraphics[keepaspectratio, width=\textwidth]{#1}
        \subcaption*{\footnotesize #3\\#2}
    \end{minipage}
    \hspace{-3.5mm}
}

\begin{figure*}[t]
    \captionsetup[subfigure]{justification=centering}
        \def\imgWidthForGaussian{0.11}
        \resizebox{\textwidth}{!}{%
        \myimg{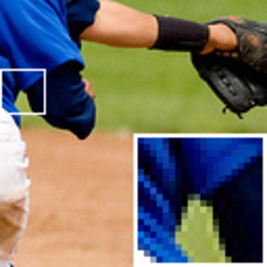}{Ground truth\vspace{3mm}}{(PSNR,\,SSIM)}{\imgWidthForGaussian}
        \myimg{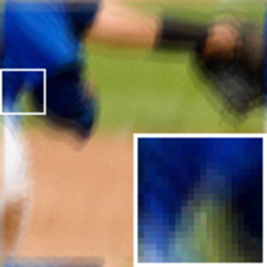}{Observation\vspace{3mm}}{$(21.29,\,0.6381)$}{\imgWidthForGaussian}
        \myimg{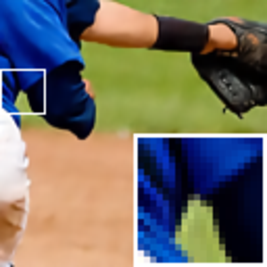}{PnP-FBS\vspace{3mm}}{$(33.47,\,0.9273)$}{\imgWidthForGaussian}
        \myimg{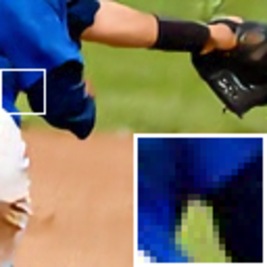}{PnP-PDS\\(Unstable)}{$(27.55,\,0.8465)$}{\imgWidthForGaussian}
        \myimg{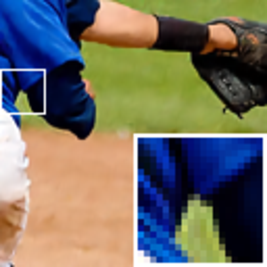}{PnP-PDS\\ (\replySec{NoBox})}{$(\underline{33.93},\,\underline{0.9313})$}{\imgWidthForGaussian}
        \myimg{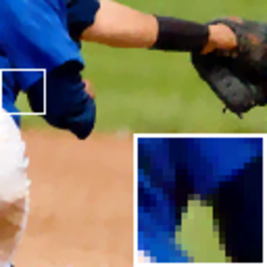}{TV\vspace{3mm}}{$(30.53,\,0.8875)$}{\imgWidthForGaussian}
        \myimg{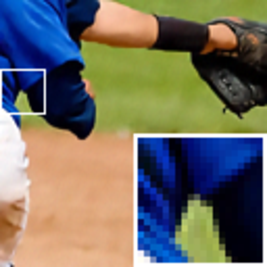}{RED\vspace{3mm}}{$(33.75,\,0.9261)$}{\imgWidthForGaussian}
        \myimg{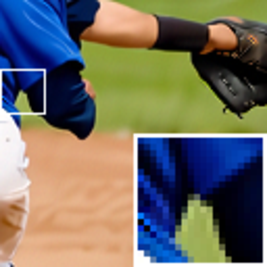}{\reply{DRUNet}\vspace{3mm}}{\reply{$(33.75,\,0.9261)$}}{\imgWidthForGaussian}
        \myimg{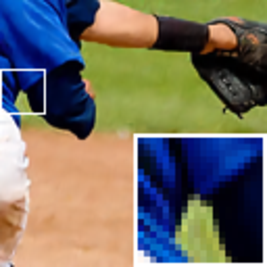}{Proposed\vspace{3mm}}{$(\mathbf{33.97},\,\mathbf{0.9317})$}{\imgWidthForGaussian}
        }
    \caption{Visual results of image restoration on ImageNet~\cite{ImageNet} for deblurring at $\sigma = 0.01$. The best results are indicated in bold, and the second-best results are underlined.}
    \label{GaussianResultVisual}
\end{figure*}

We compare the performance of the proposed method with other methods, namely PnP-FBS\cite{Pesquet2021}, PnP-PDS (Unstable)\cite{PnPPDS}, PnP-PDS (\replySec{NoBox})~\cite{PnPPDSByPesquet}, TV\cite{constraint0}, RED\cite{RED}, and DRUNet~\cite{DRUNet}. For constrained formulations such as PnP-PDS (Unstable), PnP-PDS (\replySec{NoBox}), TV, and the proposed method, we set $\varepsilon=\alpha\varepsilon_\mathrm{opt}$ in the same way as in the previous section. \reply{In particular, we provide the values of $\alpha$ for the proposed method in Table~\ref{GaussianParameterTable}. For existing methods, we adopted the recommended values if parameter settings were suggested by the authors. Otherwise, we performed a linear search over the balancing parameter.
}

Table~\ref{GaussianResultTablePerformance} summarizes the average PSNR values over seven images, obtained with the optimal parameters for each method. The proposed method demonstrates superior performance compared to PnP-FBS, even though both methods employ the same denoiser for regularization. \reply{This is attributed to the difference in the versatility of the denoiser, as mentioned in Section \ref{ComparisonOfParameterRobustness}.}
\replyThird{Moreover, we observe that DRUNet (8 iter) achieves higher PSNR values than the proposed method in deblurring at $\sigma=0.005$ and $\sigma=0.01$.} However, the proposed method outperforms DRUNet for all cases in larger iterations because of the difference in stability of the algorithms\replyThird{, which we further examine later with their convergence curves}. This instability can be attributed to the lack of a theoretical convergence guarantee, and this observation is consistent with the literature, for instance in~\cite{Terris_2024_CVPR}.

In Table~\ref{GaussianResultTablePerformance}, we also note that PnP-FBS, PnP-PDS (\replySec{NoBox}), and RED fail to converge, and thus PSNR values cannot be obtained for the inpainting. \reply{To clarify this point, recall that the denoiser was trained with the penalty in~\eqref{PenaltyOfDNCNN}, which was designed to promote nonexpansiveness for low-noise images sampled according to~\eqref{Qselection}. During the PnP iterations, however, the inputs to the denoiser can become highly noisy and may have characteristics that differ from those of the training images. In such cases, the denoiser may no longer maintain its firm nonexpansiveness, potentially amplifying instability through successive iterations.}

For numerical verification, we present the evolution of $\|\nabla Q(\mathbf{u}_n)\|_{\mathrm{sp}}^2$ at each iteration in Fig.~\ref{GaussianBoxConstraintStudy}. Although PnP-PDS (NoBox), PnP-FBS, RED, and the proposed method employ the same denoiser, the values of $\|\nabla Q(\mathbf{u}_n)\|_{\mathrm{sp}}^2$ in PnP-PDS (NoBox), PnP-FBS, and RED exhibit unstable behavior and tend to diverge. In contrast, the proposed method consistently maintains values of $\|\nabla Q(\mathbf{u}_n)\|_{\mathrm{sp}}^2$ that satisfy the condition given in~\eqref{NonexpansivenessAndJacobian}. \replySec{In Fig.~\ref{DivergenceVisualize}, we further compare the visualization of $\mathbf{u}_n$ obtained by PnP-PDS (NoBox) in a diverging case with that of the proposed method. For PnP-PDS (NoBox), erratic regions appear in the image, and the surrounding areas are affected in the proceeding iterations, leading to divergence of $\mathbf{u}_n$. Although similar irregular regions are observed at $\mathbf{u}_{753}$ in the proposed method, $\mathbf{u}_n$ remains stable and those regions are effectively removed in the subsequent iterations.}

\replyThird{To study the impact of enforcing firm nonexpansiveness on the denoiser, we compare the stability of the proposed method with that of PnP-PDS (Unstable) and DRUNet, whose denoisers are not guaranteed to be firmly nonexpansive.} Fig.~\ref{GaussianResultGraph} shows the evolution of PSNR and the update rate $c_n$ for the deblurring task with $\sigma = 0.01$. \replyThird{PnP-PDS (Unstable) exhibits unstable behavior throughout the iterations, and DRUNet shows similar behavior, becoming unstable in the later iterations.} In contrast, the proposed method converges stably, with $c_n$ decreasing to approximately $10^{-5}$.

In Fig. \ref{GaussianResultVisual}, we provide the visual results for deblurring at $\sigma = 0.01$, accompanied by PSNR (left) and SSIM (right). The images obtained by PnP-FBS and TV are overly-smoothed, resulting in a loss of fine detail. \reply{For the PnP methods without theoretical convergence guarantees, namely PnP-PDS (Unstable) and DRUNet, strong artifacts are generated because of the unstable behavior of the algorithm (see for example the area above the glove and the texture on the ground in the image, respectively).} While RED achieves high PSNR and SSIM, there are some areas where the brightness appears unnaturally altered compared to the original image. The proposed method achieves the highest PSNR and SSIM for this image, presenting the most natural appearance.

\subsection {Deblurring/Inpainting under Poisson Noise}
\label{PoissonExperiment}

\begin{table*}[t]
\centering
\setlength{\tabcolsep}{14pt}
\resizebox{\linewidth}{!}{%
\begin{tabular}{lccccccccc} \toprule
   \reply{Scaling coefficient $\eta$} & \replySec{$1$}& \replySec{$2$} &  \replySec{$10$}& \reply{$50$} & \reply{$100$} & \reply{$200$} \\ \midrule
  \reply{Deblurring} &\replySec{$2.00\times10^{-3}$}&\replySec{$2.00\times10^{-3}$} &\replySec{$1.50\times10^{-3}$} & \reply{$1.25\times10^{-3}$} & \reply{$1.25\times10^{-3}$} & \reply{$1.00\times10^{-3}$} \\
\reply{Inpainting} & \replySec{$1.25\times10^{-3}$}&\replySec{$1.25\times10^{-3}$} &\replySec{$1.00\times10^{-3}$}&\reply{$7.50\times10^{-4}$} & \reply{$5.00\times10^{-4}$} & \reply{$5.00\times10^{-4}$} \\

   \bottomrule
   \end{tabular}
\caption{\reply{The Values of $\lambda$ Used in Each Setting for the Proposed Method.}}
\label{PoissonParameterTable}
}
\end{table*}

\begin{table*}[t]
\centering
\resizebox{\linewidth}{!}{%
\begin{tabular}{lcccccccccccc} \toprule
    & \multicolumn{6}{c}{Deblurring} & \multicolumn{6}{c}{Inpainting} \\ \cmidrule(lr){2-7} \cmidrule(lr){8-13}
   Scaling coefficient $\eta$ & \replySec{$1$} & \replySec{$2$} & \replySec{$10$} & $50$ & $100$ & $200$ & \replySec{$1$} & \replySec{$2$} & \replySec{$10$} & $50$  & $100$ & $200$ \\ \midrule
   PnP-ADMM~\cite{ConvergentPnPADMM} & \replySec{$10.00$} & \replySec{$7.58$} & \replySec{$\mathbf{22.95}$} & $25.16$ & $26.28$ & $26.55$ & \replySec{$7.74$} & \replySec{$9.04$} & \replySec{$25.83$} & $11.45$ & $27.87$ & $28.10$ \\
   PnP-PDS (Unstable)~\cite{PnPPDS} & \replySec{$7.07$} & \replySec{$13.86$} & \replySec{$21.06$} & $23.98$ & $25.18$ & $26.51$ & \replySec{$5.38$} & \replySec{$8.25$} & \replySec{$20.29$} &$\mathbf{29.75}$ & $\mathbf{31.90}$ & $\mathbf{33.43}$ \\
   PnP-PDS (\replySec{NoBox})~\cite{PnPPDSByPesquet}  & \replySec{$\underline{18.85}$} & \replySec{$20.08$} & \replySec{$22.85$} & $\mathbf{25.17}$ & $\underline{26.42}$ & $\mathbf{27.77}$ & \replySec{$\underline{18.09}$} & \replySec{$20.47$} & \replySec{$\underline{25.95}$} &$\mathbf{29.75}$ & $\underline{31.21}$ & $\underline{32.74}$ \\
   \reply{BPG~\cite{BregmanProximalGradient}}  & \replySec{$5.41$} & \replySec{$6.86$} & \replySec{$16.27$} & \reply{$21.76$} & \reply{$26.36$} & \reply{$27.45$}  & \replySec{$6.52$} & \replySec{$7.85$} & \replySec{$10.19$} & \reply{$11.84$} & \reply{$12.06$} & \reply{$12.18$} \\
   RED~\cite{RED}  & \replySec{$18.64$} & \replySec{$\underline{20.12}$} & \replySec{$\underline{22.91}$} & $25.13$ & $26.39$ & $27.72$ & \replySec{$18.02$} & \replySec{$\underline{20.75}$} & \replySec{$25.88$} & $29.71$ & $31.19$ & $32.69$ \\
   PnP-PDS (Proposed) & \replySec{$\mathbf{19.04}$} & \replySec{$\mathbf{20.19}$} & \replySec{$22.87$} & $\mathbf{25.17}$ & $\mathbf{26.44}$ & $\mathbf{27.77}$ & \replySec{$\mathbf{18.59}$} & \replySec{$\mathbf{21.16}$} &  \replySec{$\mathbf{25.96}$} & $\mathbf{29.75}$ & $\underline{31.21}$ & $\underline{32.74}$ \\
   \bottomrule
   \end{tabular}
\caption{The Average PSNR [dB] Values for Deblurring and Inpainting on the Three Images From Set3\cite{Pesquet2021}. The Best Value is Highlighted in Bold, and the Second Best Value is Underlined.}
\label{PoissonResultTablePerformance}
}
\end{table*}

\begin{table}[t]
\centering
\setlength{\tabcolsep}{18pt}
\resizebox{\linewidth}{!}{%
\begin{tabular}{lcc} \toprule
    & Deblurring & Inpainting \\ \midrule
   PnP-ADMM~\cite{ConvergentPnPADMM} & \replySec{$1.525$} & \replySec{$0.5558$} \\
   PnP-PDS (Proposed) & \replySec{$\mathbf{0.5972}$} & \replySec{$\mathbf{0.3153}$} \\
   \bottomrule
   \end{tabular}}
\caption{The Average CPU Time [s] per Iteration Across Scaling Coefficients on the Three Images From Set3~\cite{Pesquet2021}. The Better Result is Highlighted in Bold.}
\label{PoissonResultTableCPUTime}
\end{table}

\def\imsize{0.243}
\begin{figure}[t]
    \centering
    \captionsetup[subfigure]{justification=centering}
    \begin{minipage}[t]{\linewidth}
    \centering
    
        \myimg{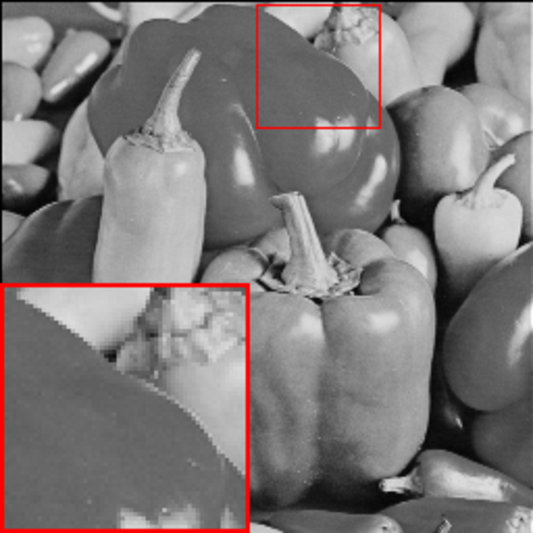}{\replySec{Ground truth}}{\replySec{(PSNR, SSIM)}}{\imsize}
        \myimg{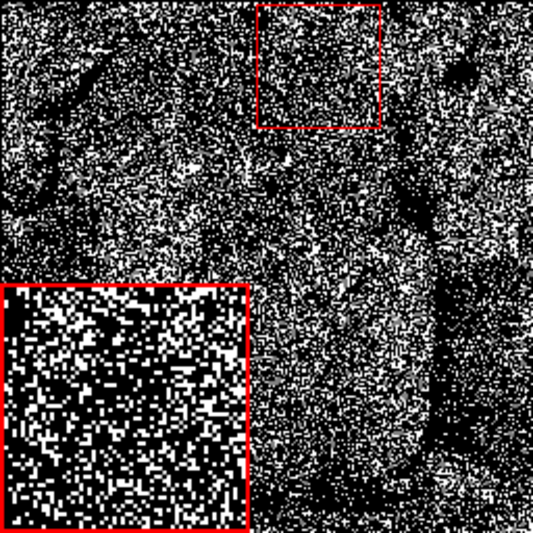}{\replySec{Observation}}{\replySec{$(3.471,\,0.0257)$}}{\imsize}
        \myimg{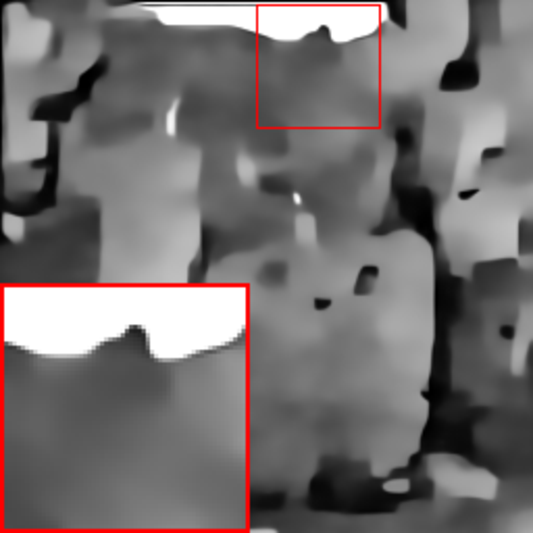}{\replySec{PnP-PDS (NoBox)}}{\replySec{$(16.64,\,0.6095)$}}{\imsize}
        \myimg{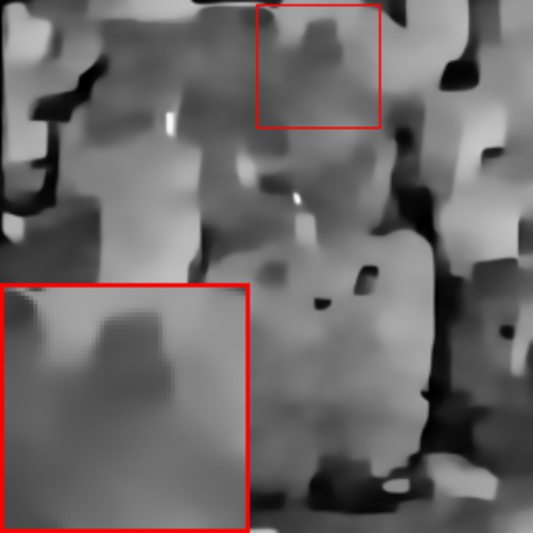}{\replySec{Proposed}}{\replySec{$(\mathbf{18.10},\,\mathbf{0.6146})$}}{\imsize}
    \end{minipage}
    \caption{\replySec{Visual results of PnP-PDS (NoBox) and the proposed method for inpainting with $\eta = 1$. The higher PSNR and SSIM values are highlighted in bold.}}
    \label{poisson_intensive_visual}
\end{figure}

\newcommand{\imgsep}{0.123}
\begin{figure*}[t]
    \captionsetup[subfigure]{justification=centering}
        \myimg{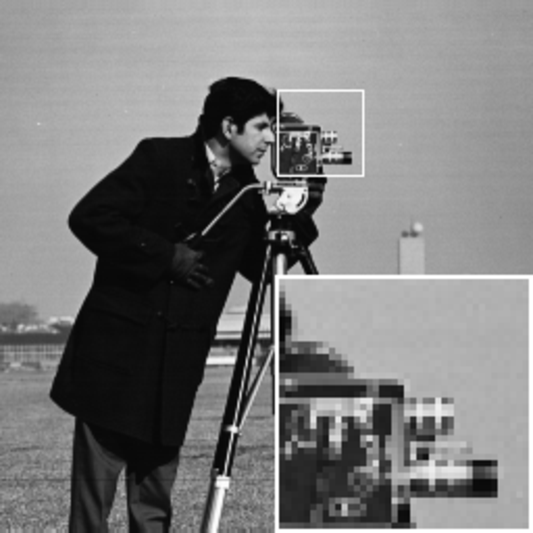}{Ground truth\vspace{3.1mm}}{(PSNR,\,SSIM)}{\imgsep}
        \myimg{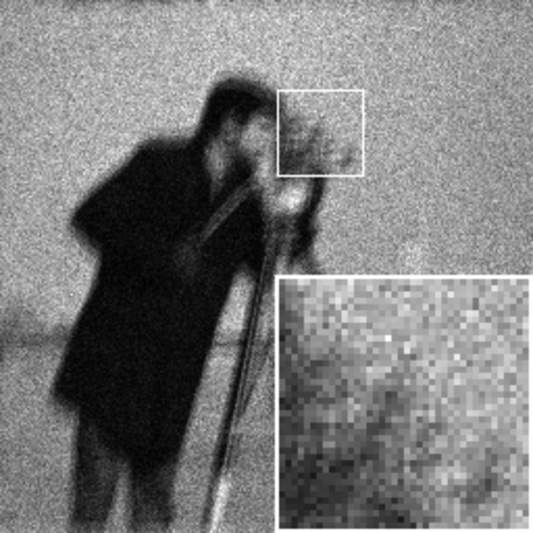}{Observation\vspace{3.1mm}}{$(19.26,\,0.2653)$}{\imgsep}
        \myimg{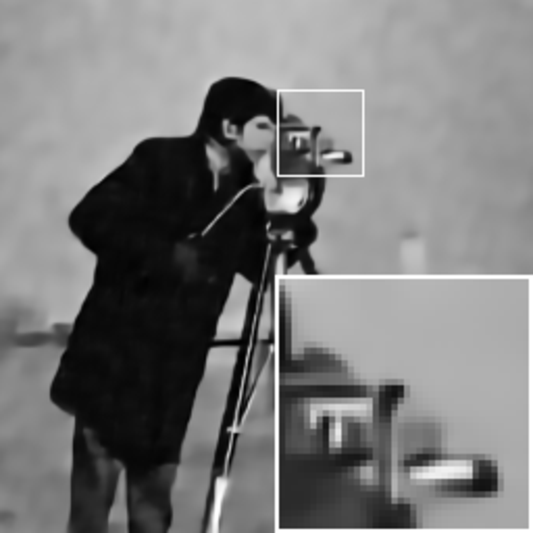}{PnP-ADMM\vspace{3.1mm}}{$(25.17,\,\underline{0.7743})$}{\imgsep}
        \myimg{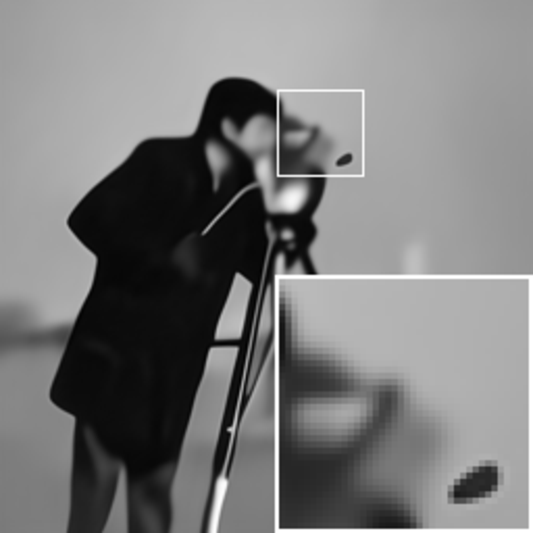}{PnP-PDS\\ (Unstable)}{$(24.05,\,0.7618)$}{\imgsep}
        \myimg{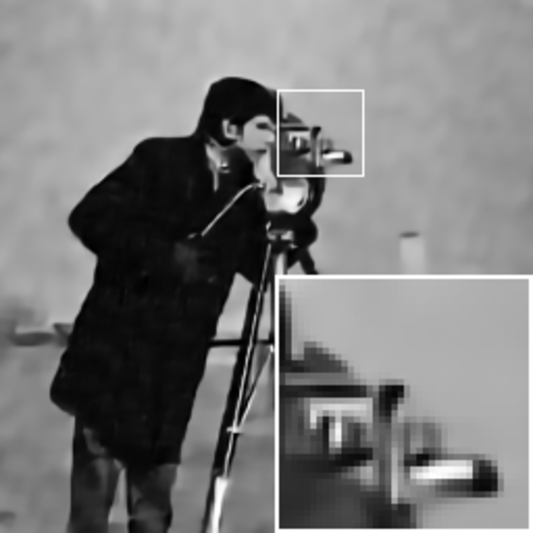}{PnP-PDS\\ (\replySec{NoBox})}{$(\underline{25.31},\,0.7731)$}{\imgsep}
        \myimg{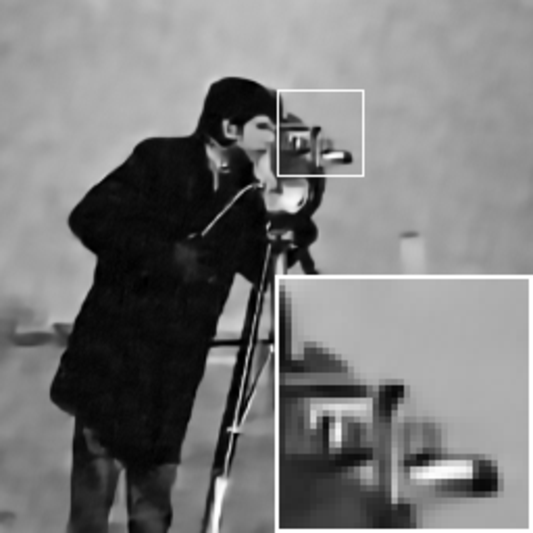}{RED\vspace{3.1mm}}{$(\underline{25.31},\,0.7678)$}{\imgsep}
        \myimg{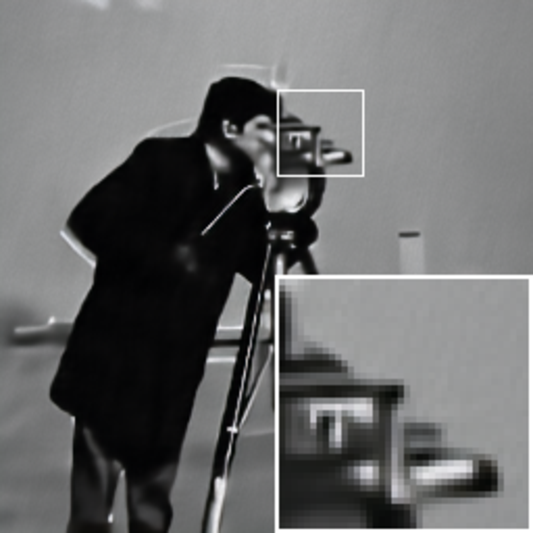}{\reply{BPG}\vspace{3.1mm}}{\reply{$(24.74,\,\mathbf{0.7829})$}}{\imgsep}
        \myimg{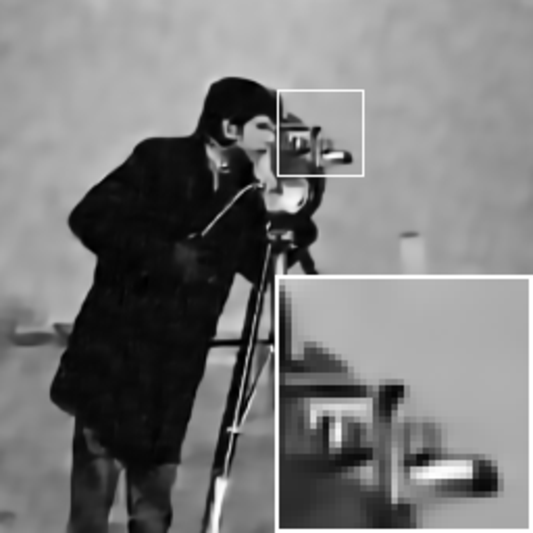}{Proposed\vspace{3.1mm}}{$(\mathbf{25.33},\,0.7733)$}{\imgsep}
    \caption{Visual results of deblurring with $\eta = 100$. The best results are indicated in bold, and the second-best results are underlined.}
    \label{PoissonResultVisual}
\end{figure*}

\begin{figure*}[t]
    \centering
    \includegraphics[keepaspectratio, width=\linewidth]{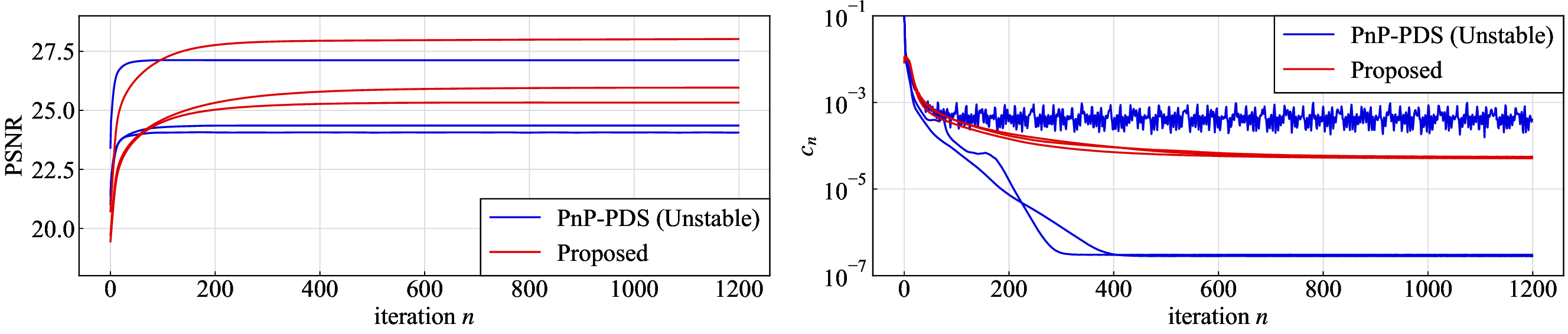}
        \caption{The evolution of PSNR~[dB] and $c_n$ of PnP-PDS (Unstable) and the proposed method over iterations for deblurring at $\eta=100$. Each plot in the graph corresponds to each image extracted from Set3\cite{Pesquet2021}.}
    \label{PoissonEvolutionGraph}
\end{figure*}

Let us now present the experiments on deblurring/inpainting under Poisson noise. \replySec{To vary the noise intensity, we considered six different values for the scaling coefficients $\eta$ between $1$ and $200$.} Smaller values of $\eta$ correspond to stronger noise, while larger values indicate weaker noise. We compare the performance of the proposed method with other methods applicable to this scenario, specifically PnP-ADMM~\cite{ConvergentPnPADMM}, PnP-PDS (Unstable)~\cite{PnPPDS}, PnP-PDS (\replySec{NoBox})~\cite{PnPPDSByPesquet}, RED~\cite{RED}\reply{, and BPG~\cite{BregmanProximalGradient}}. For methods whose parameter settings are suggested by the authors, we followed their recommendations. For the others, we performed a linear search over the balancing parameters between the data-fidelity and regularization terms. \reply{In particular, we report the best values of $\lambda$ for the proposed method in Table~\ref{PoissonParameterTable}.}

The average PSNR values over three images are presented in Table~\ref{PoissonResultTablePerformance}. The proposed method shows high restoration performance across all settings \reply{and outperforms the other methods with higher PSNR values.} \replySec{Especially, we observe that, at high noise levels (i.e., $\eta\in\{1,\,2\}$), the proposed method outperforms PnP-PDS (NoBox). This indicates that the box constraint not only enhances stability as shown in Section~\ref{GaussianExperiment}, but also contributes positively to performance under severely degraded conditions.} \replySec{In Fig.~\ref{poisson_intensive_visual}, we compare the visual results in such cases. It can be observed that PnP-PDS (NoBox) tends to generate unpleasant artifacts when the noise level becomes high, whereas the proposed method effectively suppresses such artifacts.}

\reply{
\begin{rem}[Comparison with Clamping of Denoiser Output]
An alternative approach to enforcing the box constraint is to clamp the denoiser output to $[0,\,1]$ after each iteration. This simple post-processing step can help stabilize the algorithm, and using it instead of the box constraint achieves almost the same performance as the proposed method in all settings. On the other hand, the proposed formulation explicitly incorporates the box constraint into the monotone inclusion formulation. This allows for a clearer characterization of the optimality condition and improves the interpretability of the solution.
\end{rem}
}

We investigate the computational efficiency of the proposed method, particularly in comparison with PnP-ADMM. Table~\ref{PoissonResultTableCPUTime} provides the CPU times per iteration for PnP-ADMM and the proposed method. As discussed in Section~\ref{ComparisonWithPnPADMM}, we need to perform inner iterations for PnP-ADMM, leading to relatively high computational costs. The proposed method reduces CPU times to approximately half of those of PnP-ADMM in all settings, since it solves the problem in~\eqref{PoissonPDSFormulation} without matrix inversions.

To evaluate stability, we present the PSNR and the update rate $c_n$ at each iteration in the proposed method and PnP-PDS (Unstable) in Fig.~\ref{PoissonEvolutionGraph}. For PnP-PDS (Unstable), $c_n$ behaves unstably depending on the image. In contrast, the proposed method shows a consistent decrease in $c_n$ over iterations, along with stable PSNR values.

Fig.~\ref{PoissonResultVisual} shows the visual results of the deblurring at $\eta = 100$. PnP-PDS (Unstable) loses the details of the original images, while PnP-ADMM and RED show slight changes in brightness or fine shapes. 
\reply{Although BPG preserves the structure of the image well, it also produced several visually unpleasant artifacts (e.g., above the man's shoulder and head).} 
In contrast, the proposed method recovers fine details, yielding a visually natural image. 

\subsection{\reply{Influence of the Prior on the Proposed Method}}
\label{PriorExperiment}

\begin{figure*}[t]
    \centering
    \includegraphics[keepaspectratio, width=\linewidth]{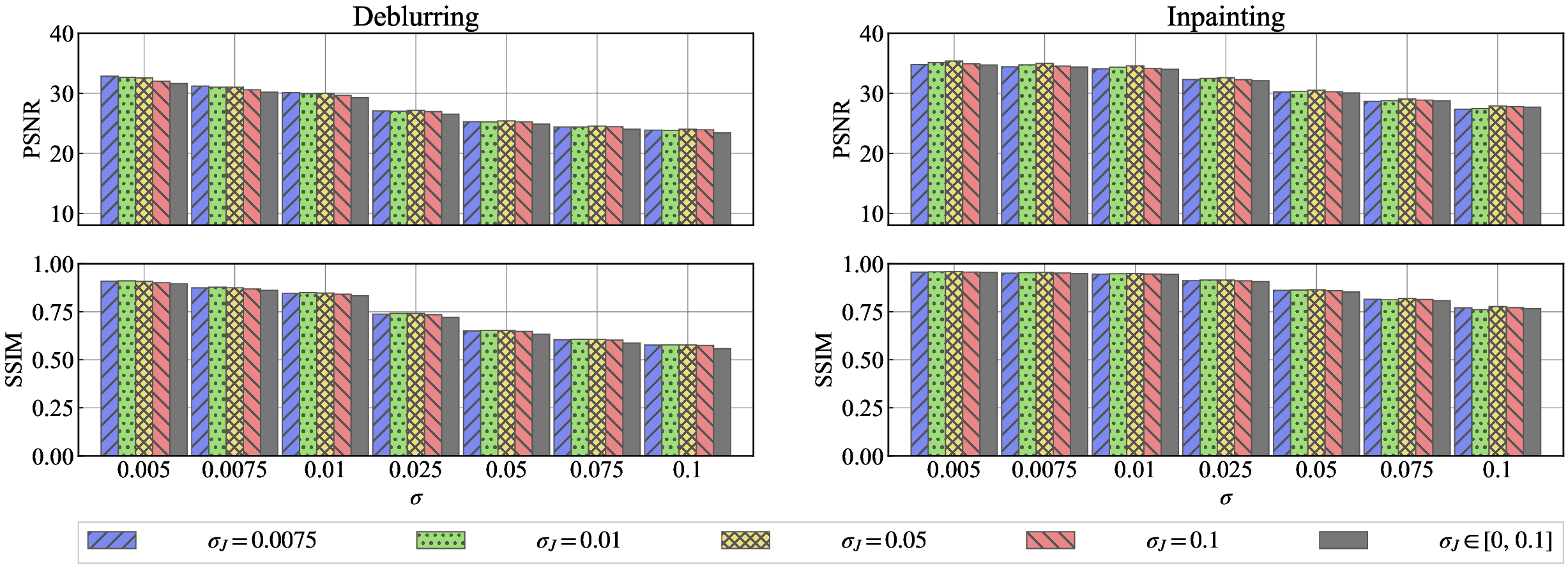}
    \caption{\reply{The influence of the variation in $\sigma_J$ on the restoration performance for the proposed method. The plotted values are the average of seven images from ImageNet~\cite{ImageNet}.}}
    \label{sigmaJ_Performance}
\end{figure*}

\reply{
We study how the choice of prior influences the performance of the proposed method. As a representative case, we focus on how the noise level $\sigma_J$ assumed during the training of the Gaussian denoiser affects reconstruction performance of deblurring/inpainting under Gaussian noise.}

\reply{We trained denoisers with four fixed noise levels, specifically $\sigma_J \in \{0.0075,\, 0.01,\, 0.05,\, 0.1\}$ where $\sigma_J$ denotes the standard deviation of the Gaussian noise added to inputs during denoiser training (i.e., in~\eqref{PenaltyOfDNCNN}, $\mathbf{x}_\ell$ is corrupted with Gaussian noise of standard deviation $\sigma_J$). In addition, we trained a mixed-noise-level model, where $\sigma_J$ was randomly sampled from the range $[0,\, 0.1]$ for each image. In all cases, the values of $\tau$ and $\xi$ in~\eqref{PenaltyOfDNCNN} were manually tuned. For more details, we follow the training procedure described in \cite{Pesquet2021}.
}

\reply{
The performance comparison is shown in Fig.~\ref{sigmaJ_Performance}. The noise level $\sigma$ of the observation was selected from seven values between $0.005$ and $0.1$. For each fixed $\sigma$, the proposed method exhibits nearly identical performance across all values of $\sigma_J$, with the maximum PSNR and SSIM differences being $1.215$ dB (in deblurring at $\sigma = 0.005$) and $0.021$ (in deblurring at $\sigma = 0.025$), respectively. These results indicate that the proposed method can effectively leverage the prior of denoisers trained with a wide range of $\sigma_J$. Furthermore, it achieves similarly competitive performance even when using a denoiser trained with randomly varying noise levels $\sigma_J \in [0, 0.1]$, suggesting that its regularization effect is not highly sensitive to the specific training noise level in the proposed method.
}

\reply{
Finally, we note that the consistent performance regardless of $\sigma_J$ and $\sigma$ can be attributed to expressing the data-fidelity as a constraint. By formulating it as a constraint, the consistency between the reconstruction and the observation can be explicitly controlled via the parameter $\varepsilon$. This helps prevent overly aggressive denoising that could arise when $\sigma_J$ is much larger than $\sigma$.
}

\section{Conclusion}
In this paper, we have proposed a general PnP-PDS method with a theoretical convergence guarantee under realistic assumptions in real-world settings, supported by extensive experimental results. First, we established the convergent PnP-PDS using a firmly nonexpansive denoiser and theoretically characterized its solution set. Then, we showed that it efficiently solves various image restoration problems involving nonsmooth data-fidelity terms and additional hard constraints without the need to compute matrix inversions or perform inner iterations. Finally, through numerical experiments on deblurring/inpainting under Gaussian noise and deblurring/inpainting under Poisson noise, we demonstrated that the proposed PnP-PDS outperforms existing methods with stable convergence behavior. We believe that the proposed PnP-PDS holds great potential for a wide range of applications including biomedical imaging, astronomical imaging, and remote sensing. 

\reply{
\appendix[Proofs of Proposition~\ref{propOfConvergence} and Proposition~\ref{propOfConvergenceWhenBetaIsZero}]}
%Now, we provide the proofs of Proposition~\ref{propOfConvergence} and Proposition~\ref{propOfConvergenceWhenBetaIsZero}, respectively. 
%\label{AppendixOfProofs}
We start by introducing several key facts, followed by the proofs.

\begin{fact}[{\hspace{-0.1mm}\cite[Lemma 4.4]{PDS2}}]
\label{LemmaOfFBS}
Let $B_1:\mathcal{H}\to2^\mathcal{H}$ be a maximally monotone operator, and $B_2:\mathcal{H}\to\mathcal{H}$ be a $\kappa$-cocoercive operator for some $\kappa\in(0,\,\infty)$. Suppose that $\mathrm{zer}(B_1+B_2)$ is nonempty, and the following conditions hold:
    \begin{enumerate}[(i)]
    \item $\gamma\in\left(0,\,2\kappa\right)$,
    \item $\rho_n\in\left(0,\,\delta\right)$,
    \item $\sum_{n\in\mathbb{N}}\rho_n(\delta-\rho_n)=\infty$,
    \end{enumerate}
where $\delta:=2-\gamma\,(2\kappa)^{-1}$.
%is defined as follows:
%\begin{align}
%\delta:=2-\dfrac{\gamma}{2\kappa}.
%\end{align}
Consider the sequence $\{\mathbf{s}_n\}_{n\in\mathbb{N}}$ generated by the following update:
\begin{align}
\label{FBSUpdate}
    \begin{split}
        \mathbf{s}_{n+1}=\rho_n((\mathrm{Id}+\gamma B_1)^{-1}(\mathrm{Id}-\gamma B_2)(\mathbf{s}_n)&)+\\
        &(1-\rho_n)\mathbf{s}_n.
    \end{split}
\end{align}
The sequence $\{\mathbf{s}_n\}_{n\in\mathbb{N}}$ generated by~\eqref{FBSUpdate} weakly converges to $\hat{\mathbf{s}}$ s.t. $\hat{\mathbf{s}}\in\mathrm{zer}(B_1+B_2)$.
\end{fact}

\begin{fact}[{\hspace{-0.1mm}\cite[Lemma 4.2]{PDS2}}]
\label{LemmaOfProximalPoint}
Let $B:\mathcal{H}\to2^\mathcal{H}$ be a maximally monotone operator. Suppose that $\mathrm{zer}(B)$ is nonempty, and the following conditions hold:
    \begin{enumerate}[(i)]
    \item $\rho_n\in\left(0,\,2\right)$,
    \item $\sum_{n\in\mathbb{N}}\rho_n(2-\rho_n)=\infty$.
    \end{enumerate}
Consider the sequence $\{\mathbf{s}_n\}_{n\in\mathbb{N}}$ generated by the following update:
\begin{align}
\label{ProximalPointUpdate}
    \mathbf{s}_{n+1}=\rho_n((\mathrm{Id}+B)^{-1}(\mathbf{s}_n)&)+(1-\rho_n)\mathbf{s}_n.
\end{align}
The sequence $\{\mathbf{s}_n\}_{n\in\mathbb{N}}$ generated by~\eqref{ProximalPointUpdate} weakly converges to $\hat{\mathbf{s}}$ s.t. $\hat{\mathbf{s}}\in\mathrm{zer}(B)$.
\end{fact}

\begin{fact}[{\hspace{-0.1mm}\cite[Lemma 4.5]{PDS2}}]
\label{LemmaOfNonexpasiveness}
Let $J:\mathcal{H}\to\mathbb{R}$ be a convex differentiable function. If $\kappa\nabla J$ is nonexpansive for some $\kappa\in(0,\,\infty)$, then $\kappa\nabla  J$ becomes firmly nonexpansive.
\end{fact}

\begin{proof}[Proof of Proposition \ref{propOfConvergence}]
We follow the proof structure as in~\cite[Theorem 3.1 for Algorithm 3.1]{PDS2}. Consider the real Hilbert space 
$\mathcal{Z}_I=\mathcal{X}\times\mathcal{Y}$ equipped with the inner product defined as $\ip{\mathbf{z}}{\mathbf{z'}}_I:=\ip{\mathbf{x}}{\mathbf{x'}}+\ip{\mathbf{y}}{\mathbf{y'}}$
%\begin{align}
%\ip{\mathbf{z}}{\mathbf{z'}}_I:=\ip{\mathbf{x}}{\mathbf{x'}}+\ip{\mathbf{y}}{\mathbf{y'}},
%\end{align}
where $\mathbf{z}=(\mathbf{x},\,\mathbf{y}),\,\mathbf{z'}=(\mathbf{x'},\,\mathbf{y'})$. Now, we define $P$ as 
\begin{align}
P:\begin{pmatrix}\mathbf{x}\\\mathbf{y}\end{pmatrix}\to
\begin{pmatrix}
    \mathrm{Id}/\gamma_1 & -\mathbf{L}^{*}\\
    -\mathbf{L} &\mathrm{Id}/\gamma_2
\end{pmatrix}
\begin{pmatrix}\mathbf{x}\\\mathbf{y}\end{pmatrix}.
\end{align}
Then, as shown in~\cite{PDS2}, $P$ is bounded, self-adjoint, and strictly positive. Hence, we obtain the real Hilbert space 
$\mathcal{Z}_P=\mathcal{X}\times\mathcal{Y}$ equipped with the following inner product: %$\ip{\mathbf{z}}{\mathbf{z'}}_P:=\ip{\mathbf{z}}{P\mathbf{z'}}_I$, 
\begin{align}
\ip{\mathbf{z}}{\mathbf{z'}}_P:=\ip{\mathbf{z}}{P\mathbf{z'}}_I,
\end{align}
Since convergence in $\mathcal{Z}_I$ and $\mathcal{Z}_P$ are equivalent~\cite[Theorem 3.1 for Algorithm 3.1]{PDS2}, we show that the sequence generated by the iteration in~\eqref{PnP-PDSIteration} weakly converges to the solution to~\eqref{PnP-PDSSolutionSet} in $\mathcal{Z}_P$.

For simple representation of Algorithm~\eqref{PnP-PDSIteration}, let us define operators $T_1$ and $T_2$ as follows:
\begin{align}
T_1(\mathbf{z}_{n}):=\begin{pmatrix}
    \gamma_1^{-1}A_J(\mathbf{x}_{n})+\mathbf{L}^{*}\mathbf{y}_{n}\\
    -\mathbf{L}\mathbf{x}_{n}+\partial h^{*}(\mathbf{y}_{n})
\end{pmatrix},
T_2(\mathbf{z}_n):=\begin{pmatrix}
    \nabla f(\mathbf{x}_n)\\0
\end{pmatrix},\hspace{-2mm}
\end{align}
where $\mathbf{z}_{n}:=(\mathbf{x}_n^\top,\,\mathbf{y}_n^\top)^{\top}$. Then, we rewrite Algorithm~\eqref{PnP-PDSIteration} as
\begin{align}
    \label{PnPPDSUpdateFBSVer}
    \mathbf{z}_{n+1}=\rho_n(\mathrm{Id}+M_1)^{-1}(\mathrm{Id}-M_2)\mathbf{z}_n+(1-\rho_n)\mathbf{z}_n,
\end{align}
where $M_1:=P^{-1}\circ T_1$ and $M_2:=P^{-1}\circ T_2$. Since the update in~\eqref{PnPPDSUpdateFBSVer} has the same form as~\eqref{FBSUpdate} with $\gamma=1$, it is sufficient to demonstrate that each assumption in Fact \ref{LemmaOfFBS} is satisfied. Hereunder, we define $\kappa$ as follows:
\begin{align}
    \label{definitionOfKappa}
    \kappa:=\dfrac{1}{\beta}\left(\dfrac{1}{\gamma_1}-\gamma_2\|\mathbf{L}\|_{\mathrm{op}}^2\right).
\end{align}

\begin{itemize}
\setlength{\leftskip}{-10pt}
    \item We see that the set-valued operator $A_J=J^{-1}-\mathrm{Id}$ is maximally monotone from \cite[Proposition 23.8 (iii)]{ConvexBook}. Hence, $\gamma_1^{-1}A_J$ is also maximally monotone \cite[Proposition 20.22]{ConvexBook}. The operator $\partial h^{*}$ is maximally monotone from \cite[Theorem 20.48, Proposition 20.22, and Corollary 16.30]{ConvexBook}. Therefore, the operator $(\mathbf{x},\mathbf{y})\mapsto \gamma_1^{-1}A_J(\mathbf{x})\times \partial h^{*}(\mathbf{y})$ is maximally monotone in $\mathcal{Z}_I$ from \cite[Proposition 20.23]{ConvexBook}. Furthermore, according to \cite[Example 20.35]{ConvexBook}, $(\mathbf{x},\mathbf{y})\mapsto \mathbf{L}^{*}\mathbf{y}\times-\mathbf{L}\mathbf{x}$ is maximally monotone and has a full domain. Therefore, $T_1$ becomes a maximally monotone operator on $\mathcal{Z}_I$ \cite[Corollary 25.5 (i)]{ConvexBook}. Since $P$ is injective, $M_1=P^{-1}\circ T_1$ is maximally monotone on $\mathcal{Z}_P$.
    \item From Fact~\ref{LemmaOfNonexpasiveness} and the computation presented in \cite[Theorem 3.1 for Algorithm 3.1]{PDS2}, $M_2$ is cocoercive, i.e., $\kappa M_2$ is firmly nonexpansive in $\mathcal{Z}_P$.
    \item We set $\gamma=1$. Since $0<2\kappa<1$ from (i) of Proposition \ref{propOfConvergence}, the assumption (i) of Fact \ref{LemmaOfFBS} is satisfied. 
    \item The assumption (ii) of Fact \ref{LemmaOfFBS} is satisfied from~\eqref{definitionOfKappa}.
    \item The set of solutions to~\eqref{PnP-PDSSolutionSet} coincides $\mathrm{zer}(T_1+T_2)=\mathrm{zer}(M_1+M_2)$, and this set is nonempty by the assumption. 
 \end{itemize}

Therefore, by applying Fact \ref{LemmaOfFBS}, we have that $\mathbf{z}_n$ weakly converges to $\mathbf{\hat{z}}\in\mathrm{zer}(T_1+T_2)$ on $\mathcal{Z}_P$, which implies $\{\mathbf{x}_n,\mathbf{y}_n\}_{n\in\mathbb{N}}$ weakly converges to a solution to~\eqref{PnP-PDSSolutionSet}.
\end{proof}

\begin{proof}[Proof of Proposition \ref{propOfConvergenceWhenBetaIsZero}]
Following the definition of $P$ as in Proposition 3.1, $P$ is bounded, linear, and self-adjoint. Since condition (i) does not include equality, $P$ is strictly positive. Therefore, same as the case in Proposition 3.1, we can define a Hilbert space $\mathcal{Z}_P$ equipped with the inner product $\langle \cdot|\cdot \rangle_P$. Since $f(\mathbf{x})=0$ , we have $T_2=0$, and the update in~\eqref{PnPPDSUpdateFBSVer} becomes
\begin{align}
    \label{PnPPDSUPdateFBSVerWhenBetaIsZero}
    \mathbf{z}_{n+1}=\rho_n(\mathrm{Id}+M_1)^{-1}\mathbf{z}_n+(1-\rho_n)\mathbf{z}_n.
\end{align}
The update in~\eqref{PnPPDSUPdateFBSVerWhenBetaIsZero} is the same form as Fact \ref{LemmaOfProximalPoint} by setting $B=M_1$, so we need to show that the assumptions of Fact \ref{LemmaOfProximalPoint} are satisfied. Same as the proof of Proposition \ref{propOfConvergence}, $M_1$ becomes a maximally monotone operator, and from the assumption, $\mathrm{zer}(T_1 + T_2) = \mathrm{zer}(T_1) = \mathrm{zer}(M_1)$ is nonempty. Furthermore, the assumption on $\rho_n$ in Proposition \ref{propOfConvergence} is the same as in Fact \ref{LemmaOfProximalPoint}. Therefore, by Fact \ref{LemmaOfProximalPoint}, $\mathbf{z_n}$ weakly converges to $\mathrm{zer}(T_1 + T_2)$, which means $\{\mathbf{x}
_n,\,\mathbf{y}_n\}$ weakly converges to the solution to~\eqref{PnP-PDSSolutionSet}.
\end{proof}

% or


%\appendices


%\section*{Acknowledgment}
%This work was supported in part by JST PRESTO under Grant JPMJPR21C4, JST AdCORP under Grant JPMJKB2307, JST ACT-X under Grant JPMJAX24C1, JST BOOST, Japan Grant Number JPMJBS2417, and in part by JSPS KAKENHI under Grant 22H03610, 22H00512, 23H01415, 23K17461, 24K03119, and 24K22291.

\ifCLASSOPTIONcaptionsoff
  \newpage
\fi



%\bibliographystyle{IEEEtran}
\bibliographystyle{IEEEtran}

\setlength{\parskip}{0pt}
\begin{IEEEbiography}[{\includegraphics[width=1in,height=1.25in,clip,keepaspectratio]{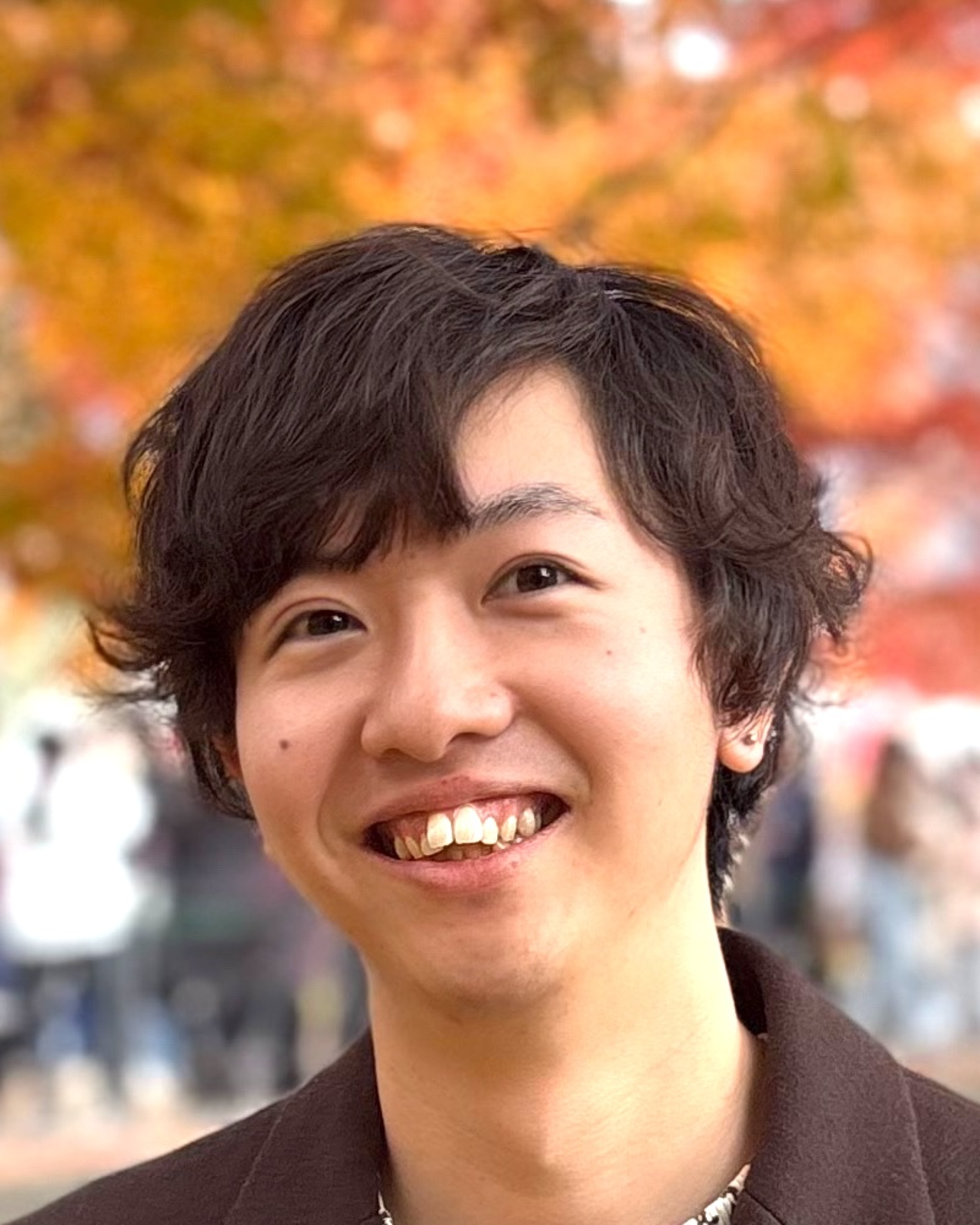}}]{Yodai Suzuki}
(S’24) received a B.E. degree in Information and Computer Science in 2024 from the Tokyo Institute of Technology. He is currently pursuing a master's degree with the Department of Computer Science at the Institute of Science Tokyo. His current research interests include signal and image processing, mathematical optimization, neural networks, and convex analysis. He received the Young Researchers' Award from the IEICE and the PCSJ/IMPS Student Paper Award in 2025.
\end{IEEEbiography}
\begin{IEEEbiography}[{\includegraphics[width=1in,height=1.25in,clip,keepaspectratio]{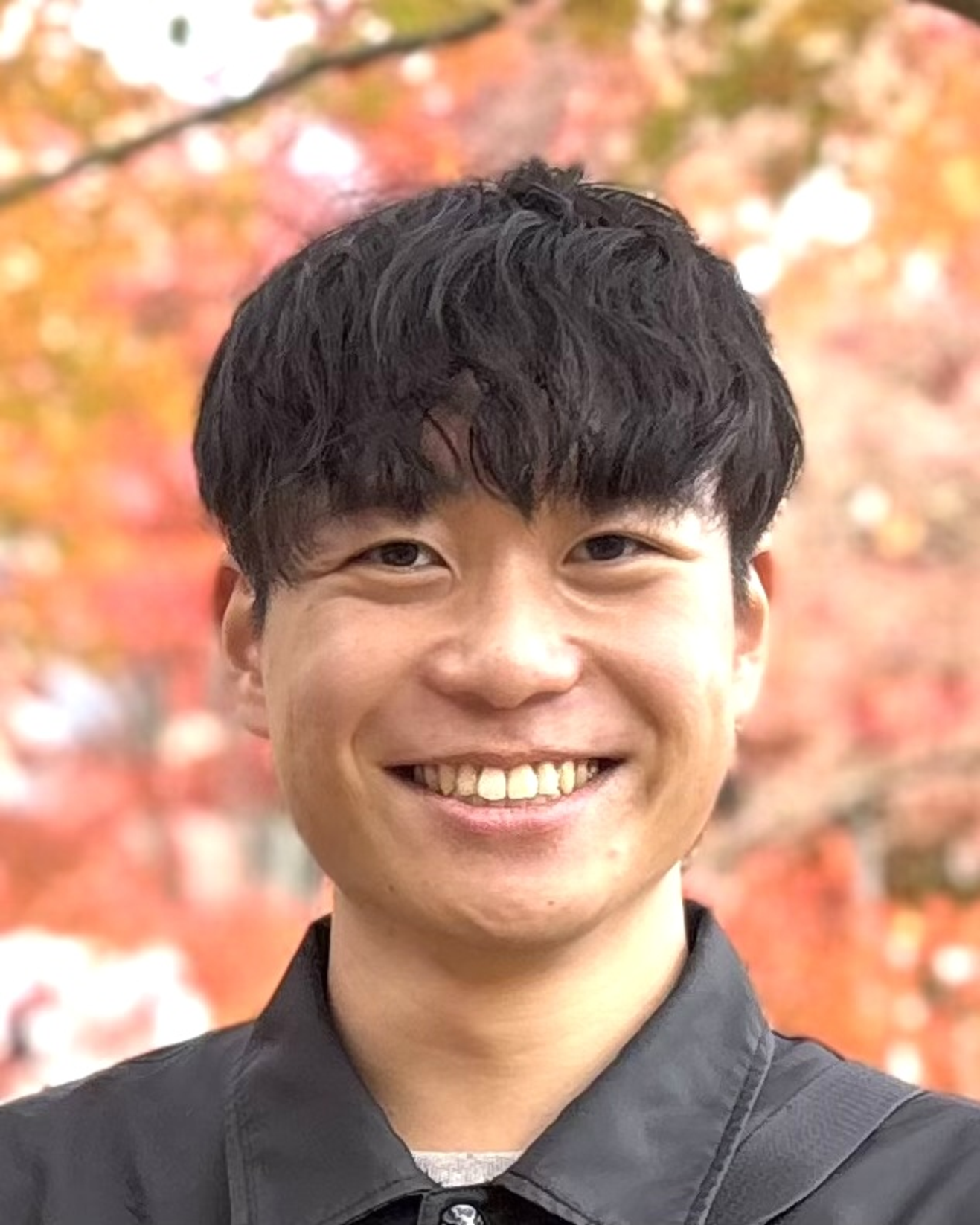}}]{Ryosuke Isono}
(S’23) received B.E. and M.E. degrees in Information and Computer Science in 2022 from the Osaka University and from the Tokyo Institute of Technology, respectively. He is currently pursuing a Ph.D. degree with the Department of Computer Science at the Institute of Science Tokyo. His current research interests include signal and image processing, mathematical optimization, and remote sensing. He received the PCSJ/IMPS Best Poster Award in 2024.
\end{IEEEbiography}
\begin{IEEEbiography}[{\includegraphics[width=1in,height=1.25in,clip,keepaspectratio]{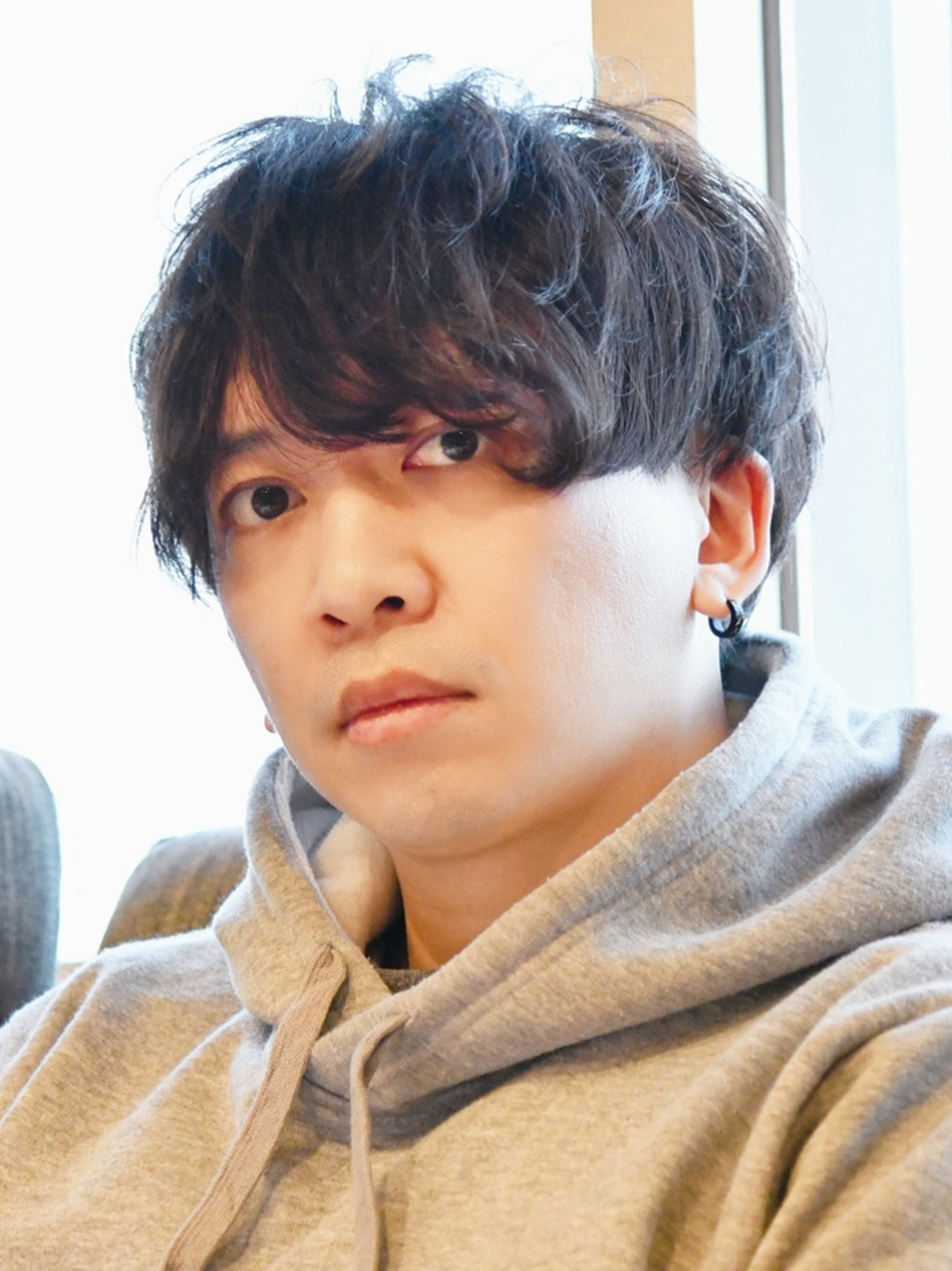}}]{Shunsuke Ono}
(S’11–M’15–SM'23) received a B.E. degree in Computer Science in 2010 and M.E. and Ph.D. degrees in Communications and Computer Engineering in 2012 and 2014 from the Tokyo Institute of Technology, respectively. From 2012 to 2014, he was a Research Fellow (DC1) of the Japan Society for the Promotion of Science (JSPS). He was an Assistant, then an Associate Professor with Tokyo Institute of Technology (TokyoTech), Tokyo, Japan, from 2014 to 2024. From 2016 to 2020, he was a Researcher of Precursory Research for Embryonic Science and Technology (PRESTO), Japan Science and Technology Agency (JST), Tokyo, Japan. Currently, he is an Associate Professor with Institute of Science Tokyo (Science Tokyo), Tokyo, Japan. His research interests include signal processing, image analysis, optimization, remote sensing, and measurement informatics. He has served as an Associate Editor for IEEE TRANSACTIONS ON SIGNAL AND INFORMATION PROCESSING OVER NETWORKS (2019--2024). Dr. Ono was a recipient of the Young Researchers’ Award and the Excellent Paper Award from the IEICE in 2013 and 2014, respectively, the Outstanding Student Journal Paper Award and the Young Author Best Paper Award from the IEEE SPS Japan Chapter in 2014 and 2020, respectively, and the Best Paper Award in APSIPA ASC 2024. He also received the Funai Research Award in 2017, the Ando Incentive Prize in 2021, the MEXT Young Scientists’ Award in 2022, the IEEE SPS Outstanding Editorial Board Member Award in 2023, and the KDDI Foundation Award in 2025. 
\end{IEEEbiography}


%\newpage


%\vfill


\end{document}